%% file: soph-arxiv.tex
\documentclass[11pt]{article}
\usepackage{mathtools,tikz,parskip,algorithm,amssymb}
\usepackage[noend]{algpseudocode}
\usepackage{geometry}
\usepackage{times}
\usepackage[caption=false]{subfig}
\usepackage{graphicx}
\newcommand{\arxiv}[1]{#1} 
\newcommand{\EC}[1]{} 

\arxiv{
\newcommand{\shortcite}[1]{\cite{#1}}
}

\newcommand{\omt}[1]{}
\newcommand{\xhdr}[1]{ \vspace{-0.03in} \paragraph*{\bf {#1}}}
\newcommand{\yhdr}[1]{ \subsubsection{\bf {#1}}}
\def\eps{\varepsilon}

\def\squarebox#1{\hbox to #1{\hfill\vbox to #1{\vfill}}}
\newcommand{\qed}{\hspace*{\fill}\vbox{\hrule\hbox{\vrule\squarebox{.667em}\vrule}\hrule}\smallskip}
\newenvironment{proof}{\noindent{\bf Proof:~~}}{\(\qed\)}
\newcommand{\figureswraper}[1]{#1} 

\newtheorem{theorem}{Theorem}[section]

\newtheorem{claim}[theorem]{Claim}

\newtheorem{example}[theorem]{Example}

\usetikzlibrary{positioning,decorations.pathmorphing,shapes,decorations.pathreplacing}
\newcommand{\p}[1]{\left(#1\right)}

\newcommand{\argmin}{\operatornamewithlimits{argmin}}
\newcommand{\argmax}{\operatornamewithlimits{argmax}}

\begin{document}
\title{Planning Problems for Sophisticated Agents with Present Bias} 
\date{}

%
%
%
%

\author{  
Jon Kleinberg
\thanks{
Cornell University.
Email: kleinber@cs.cornell.edu.
}
 \and 
Sigal Oren
\thanks{
Ben-Gurion University of the Negev.
Email: orensi@cs.bgu.ac.il.
}
 \and 
Manish Raghavan
\thanks{
U.C. Berkeley.
Email: manishraghavan@berkeley.edu.
}
}

\begin{titlepage}
\maketitle

\begin{abstract}
\input{abstract}

\end{abstract}

\thispagestyle{empty}
\end{titlepage}

\section{Introduction}
\input{intro}


\section{Behavior of Sophisticated Agents and the Cost Ratio}
\input{cost-ratio}

\section{Motivating an Agent to Reach the Goal} \label{sec:motivating}
\input{reward-at-target}

\section{Reward Seeking Behavior}
\input{rewards}

\section{Conclusion}
\input{conclusion}

\EC{
\bibliographystyle{ACM-Reference-Format-Journals}
}

\arxiv{
\bibliographystyle{plain}
}
\bibliography{refs}

\newpage
\begin{appendix}
\section{Partially Naive Agents} \label{app:partial-naive}
\input{partial-naive}

\section{Traversable Rewards}
\input{traversable}
\EC{
\section{Reward Seeking Behavior for Naive Agents} \label{app-naive-reward}
\input{app-rewards}
}
\section{Commitment Devices} 
\input{commitment}
\end{appendix}

\end{document}

%% file: abstract.tex
{\em Present bias}, the tendency to weigh costs and benefits incurred in the present too heavily, is one of the most widespread human behavioral biases. It has also been the subject of extensive study in the behavioral economics literature. While the simplest models assume that decision-making agents are {\em naive}, reasoning about the future without taking their bias into account, there is considerable evidence that people often behave in ways that are {\em sophisticated} with respect to present bias, making plans based on the belief that they will be present-biased in the future. For example, committing to a course of action to reduce future opportunities for procrastination or overconsumption are instances of sophisticated behavior in everyday life. 

Models of sophisticated behavior have lacked an underlying formalism
that allows one to reason over the full space of multi-step tasks that
a sophisticated agent might face, and this has made it correspondingly
difficult to make comparative or worst-case statements about the
performance of sophisticated agents in arbitrary scenarios. 
In this paper, we incorporate the framework of sophistication 
into a graph-theoretic model that we used in recent work for
modeling naive agents.
This new synthesis of two formalisms --- sophistication and graph-theoretic
planning --- uncovers a rich structure that wasn't
apparent in the earlier behavioral economics work on this problem, 
including a range of findings that shed new light on sophisticated behavior.

In particular, 
our graph-theoretic model makes two kinds of new results possible. First, we give tight worst-case bounds on the performance of sophisticated agents in arbitrary multi-step tasks relative to the optimal plan, along with worst-case bounds for related questions. Second, the flexibility of our formalism makes it possible to identify new phenomena about sophisticated agents that had not been seen in prior literature: these include a surprising non-monotonic property in the use of rewards to motivate sophisticated agents; a sharp distinction in the performance of agents who overestimate versus underestimate their level of present bias; and a framework for reasoning about {\em commitment devices} that shows how certain classes of commitments can produce large gains for arbitrary tasks.

%% file: intro.tex
In many everyday situations, people have a tendency to focus too heavily
on costs and benefits that are incurred immediately, rather than
balancing them against costs and benefits in the future.
We often procrastinate, for example, because we don't feel like
doing a task in the present moment, even when there are concrete
reasons why doing it now would be a better decision than doing it later.
Or we consume or spend too much because the immediate rewards
are so salient that they outweigh greater long-term benefits.

In behavioral economics, these errors in decision-making are
collectively referred to as {\em present bias}
\cite{akerlof-procrastination,pollak-time-inconsist,strotz-time-inconsist},
and they constitute one of the most widespread human behavioral biases.
Due to its range of applicability, present bias has been the
subject of extensive study 
\cite{frederick-time-inconsist-surv,dellavigna2007psychology},
including both empirical and experimental components 
as well as theoretical models.

\xhdr{Modeling Present Bias}
We imagine an agent contemplating a sequence of decisions in which
they incur a payoff of $u_t$ at a point $t$ steps into the future.
Classical models of planning posit that 
an agent who is consistently discounting the future will value this
payoff in the present step at $\delta^t u_t$, for a fixed parameter
$\delta \leq 1$.
This form of exponential discounting has the property that in the next step,
the relative values of all payoffs look just as they did in the present step,
only shifted by a factor $\delta$.

Most current
models of present bias begin from the principle of {\em quasi-hyperbolic
discounting} \cite{laibson-quasi-hyperbolic}, which for our
purposes here can be described as follows.
The behavior of a present-biased agent is governed not just by
standard exponential discounting but also by an
additional parameter $b > 1$.
Payoffs that the agent incurs in the present step are {\em increased} by
a factor of $b$ relative to all others, whereas payoffs $t \geq 1$
steps into the future continue to be reduced by $\delta^t$.
Thus, a present-biased agent perceives the payoff $u_0$ in the present step
as $b u_0$, and it perceives the payoff $u_t$ incurred $t$ steps from now,
for $t \geq 1$, as $\delta^t u_t$.
This means that the present has a fundamentally different payoff structure
than every other step; it is consistent with a wide range of
studies showing that people view the difference between ``today''
and ``tomorrow'' in a qualitatively different way than the difference
between ``$t$ days from now'' and ``$t+1$ days from now,'' for any
$t \geq 1$
\cite{frederick-time-inconsist-surv}.

The issues that we are interested in here emerge already when $\delta = 1$ ---
i.e. when the future is not discounted, so that only the present-bias 
parameter $b > 1$ is playing a role. {In this paper we focus principally
on the case of $\delta = 1$ and $b > 1$ as it allows us to separate to
an extent the role of present bias from the effects of traditional
exponential discounting.}

\xhdr{Naivete and Sophistication}
Perhaps the fundamental distinction among present-bias models 
based on quasi-hyperbolic discounting
is the contrast between {\em naive agents} and 
{\em sophisticated agents}, an issue brought into focus by influential work of 
O'Donoghue and Rabin 
\shortcite{odonoghue-now-or-later,odonoghue-choice-procrastination}.
Naive agents weigh the present payoff more highly than future ones
(by the underlying factor of $b > 1$), 
but when planning for the future, they naively believe that they will
not suffer from this bias in future steps.
Sophisticated agents also weigh the present payoff more highly than future ones,
but they formulate plans with the understanding that
they will also experience this bias in their future decision-making.

It is useful to work through the differences between naive and
sophisticated planning --- and the contrast of both with optimal
planning --- in some simple concrete instances.
We begin with two examples adapted from 
O'Donoghue and Rabin 
\shortcite{odonoghue-now-or-later,odonoghue-choice-procrastination}.

\begin{quote}
{\em A 2-period example.}
First, suppose that a present-biased agent with bias parameter $b > 1$
can perform a task today for a cost of $1$, or they can perform it
tomorrow for a cost of $c$, where $1 < c < b$.
(Imagine for example, that they need to finish writing a report,
and are deciding on Thursday whether to do it Thursday night or Friday night;
having to work late Friday night comes at a higher cost by a factor of $c$.)
Both the naive and the sophisticated agent will reason that performing
the task today comes at a perceived cost of $b$ (since that is $b$ times 
higher than its true cost, reflecting their distaste for having to do things
now), while performing the task tomorrow comes at a perceived cost of
$c < b$ (the true cost without multiplication by $b$).
As a result, both will wait and do the task tomorrow; and in this way,
both are behaving sub-optimally, since they could achieve a cost of $1 < c$
by doing the task today.
\end{quote}

There is already a crucial subtlety in the 2-period example, 
in that the naive and sophisticated
agents are making the same sub-optimal choice but for different reasons.
The naive agent's reasoning is in a sense
simpler: they mistakenly believe that when
tomorrow comes, they will perceive the cost of performing the task as $c$.
The sophisticated agent knows that tomorrow they will perceive the cost
as $bc$, but this is not part of their optimization: they want to minimize
the cost as they perceive it in the present, which is $b$ times the cost
in the present period, plus the true cost in all future periods.

In this sense, the easiest way to think about the behavior of a sophisticated
agent is in terms of the following metaphor from the behavioral economics
literature.
Imagine that when any agent --- optimal, naive, or sophisticated ---
engages in planning over multiple time periods, the decision-making in each
period is controlled by a conceptually different copy of the agent:
in period $t$, the ``period-$t$ self'' reasons about what to do,
and then hands off control of future decision-making to the 
``period-$(t+1)$ self.''
The conceptual
value in thinking about different selves is that for the period-$t$ self,
the present is in period $t$, and that can induce a potentially different
decision problem than the one faced by the period-$(t+1)$ self,
for whom the present is in period $t+1$.

Viewed this way, the different types of agents think about
cost-minimization as follows.
\begin{itemize}
\item An optimal agent in period $t$ is simply trying to 
minimize the sum of the costs incurred by itself and all future selves
(i.e. all period-$t'$ selves for $t' \geq t$).
\item 
A naive agent in period $t$ is selfish: 
it is also minimizing the sum of costs to its present and future selves,
but through a weighted sum that values the cost
to its (present) period-$t$ self a factor of $b$ times higher.
However, it naively believes that all its future selves will reason
differently, and in particular as an optimal agent would.
So it hands off control of future decision-making with a mistaken
understanding of how this decision-making will take place.
\item Finally, a sophisticated agent in period $t$ is also selfish:
it is maximizing the same weighted sum of costs to present and future
selves that the naive agent is, weighing the cost to the present
self a factor of $b$ higher.
However, it knows that when it hands off decision-making control to
its period-$(t+1)$ self, this period-$(t+1)$ will also be a sophisticated
agent with present-bias parameter $b$.
\end{itemize}

This is the sense in which the sophisticated agent can be aware of its bias
but still behave sub-optimally: like the naive agent,
it wants to favor its present self
(by a factor of $b$) using an objective function that also includes 
the costs to future selves; but unlike the naive agent,
it is realistic about the effect of its present behavior on these costs.

The distinction between naive and sophisticated agents shows up strongly
in the second of our two examples.

\begin{quote}
{\em A 3-period example.}
Now suppose that a present-biased agent 
can perform a task Thursday for a cost of $1$, 
Friday for a cost of $c$, or Saturday for a cost of $c^2$;
and suppose we choose $c$ so that $1 < c < b < c^2$.
The naive agent believes that if it waits until Friday, then its
``Friday self'' will behave optimally and
do it right away (rather than incurring the higher cost on Saturday).
Unfortunately, when Friday comes, the naive agent prefers Saturday
to Friday, and so incurs a cost of $c^2$.
The sophisticated agent knows that its Friday self will face a version
of our previous 2-period example and will wait until Saturday, incurring
a cost of $c^2$.
Hence on Thursday, the sophisticated agent faces the following binary choice:
do the task now, for a perceived cost of $b$ (since that is $b$
times the true cost); or don't do the task now, in which case decision-making
control is handed off to the Friday self for an eventual cost of $c^2$.
Since $b < c^2$, the sophisticated agent does the task right away,
which is the optimal decision.
It is making this decision 
because it still values the cost incurred by its future selves
(just at a factor of $b$ less), and in this case the consequence of delay
on its future cost is simply too high ($c^2 > b$) to justify delaying.
\end{quote}

There is extensive evidence, experimentally and in everyday life, 
that both naive and sophisticated reasoning play a large role in
different contexts
\cite{frederick-time-inconsist-surv,dellavigna2007psychology}.
We see the effects of naive reasoning when people repeatedly push off
a task into the future, always assuming that they'll get to it soon.
We see the effects of sophisticated reasoning, for example, when people
lock themselves into savings plans that make changes costly,
or when they spend effort to eliminate distractions from their work
environment in order to reduce opportunities for procrastination.

However, the complexity of sophisticated agent behavior has meant
that the range of existing theoretical results on the problem
have in many cases required different, distinct formulations and analyses.
As a result, despite the long history of work on 
modeling sophisticated behavior theoretically,
it has been a challenge to speak in general about a sophisticated
agent facing an arbitrary multi-step task,
and to quantify over the set of all such tasks
so as to be able to make comparative statements or worst-case guarantees.

To take some stark examples of the limitations in our knowledge,
there has not been a framework to answer questions such as the following:
How bad can a sophisticated agent's performance be relative to the
optimum, for an arbitrary planning problem?
How does the worst-case performance of a naive agent compare to
the worst-case performance of a sophisticated agent?
Are there small systematic ways of modifying a task to make it more
favorable to the behavior of a sophisticated agent?

\xhdr{A Graph-Theoretic Model for Sophisticated Agents}
Our goal in the present work is to explore a framework capable 
of producing precise guarantees about the performance of sophisticated
agents in general settings.
For this purpose, we employ a graph-theoretic model of planning
with present bias introduced in recent work by 
two of the authors of the present paper 
\cite{ko-ec14}.
In this earlier work, we used it to analyze naive behavior, and
here we extend it to analyze sophisticated behavior, as well as
common generalizations of the two.
Given the major differences between naive and sophisticated behavior,
the first issue is to adapt the graph-theoretic model itself;
following this, we will see that the underlying questions about
agent behavior become fundamentally different in the naive
and sophisticated cases, and reveal a range of new phenomena.


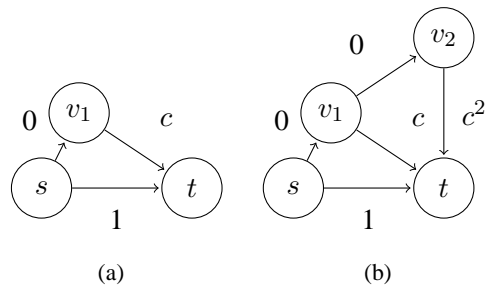
\begin{figure}[ht]
  \centering
  \subfloat[\label{fig:one-fan}]{
    \centering
    \begin{tikzpicture}[->,shorten >=1pt,auto,node distance=2cm, thin, minimum
      size=0.8cm]
      \node (s) [circle, draw] at (0,0) {$s$};
      \node (v1) [circle, draw] at (0.5,1) {$v_1$};
      \node (t) [circle, draw] at (2,0) {$t$};
      
      \path
      (s) edge node {0} (v1)
      (s) edge node [below] {1} (t)
      (v1) edge node {$c$} (t)
      ;
    \end{tikzpicture}
  }
  ~ 
  \subfloat[\label{fig:two-fan}]{
    \centering
    \begin{tikzpicture}[->,shorten >=1pt,auto,node distance=2cm, thin, minimum
      size=0.8cm]
      \node (s) [circle, draw] at (0,0) {$s$};
      \node (v1) [circle, draw] at (0.5,1) {$v_1$};
      \node (v2) [circle, draw] at (2,2) {$v_2$};
      \node (t) [circle, draw] at (2,0) {$t$};
      
      \path
      (s) edge node {0} (v1)
      (s) edge node [below] {1} (t)
      (v1) edge node {$c$} (t)
      (v1) edge node {0} (v2)
      (v2) edge node {$c^2$} (t)
      ;
    \end{tikzpicture}
  }
%
  \caption{A present-biased agent must choose a path from $s$ to $t$, where $c$ is slightly smaller than $b$.}
  \label{fig:intro-ex01}
\end{figure}

First, the model itself has the following simple but expressive structure.
There is a directed acyclic graph $G$, with costs on its edges,
representing an implicit
state space underlying an agent's set of future decisions.
The agent must travel from $s$ to $t$ in the graph $G$, and
it seeks an $s$-$t$ path that will minimize its cost in doing so.
The complication is that when the agent is at a node $v$,
the costs of all edges out of $v$ are multiplied up by a 
factor $b > 1$.
This increase in the costs out of the current node captures the
effect of present bias: the cost the agent is about to incur,
due to its upcoming step out of $v$, appears to be larger
than it did when viewed from elsewhere in the graph.

It is useful to adapt the language of present and future selves
to the graph-theoretic context.
When the agent at node $v$ decides to follow an edge to node $w$,
it is handing off decision-making control from its ``node-$v$ self''
to its ``node-$w$ self.''
The naive agent believes (mistakenly) that the node-$w$ self will
plan a minimum-cost path to the target node $t$.
The sophisticated agent correctly understands that the node-$w$ self
will also be a sophisticated present-biased agent. 
It therefore chooses
this next node $w$ with this consideration
in mind, to minimize $b$ times the
cost of the immediately next edge (the cost to its present self, scaled up
by $b$) plus the cost of all remaining edges on the path that
its future selves will collectively construct.

As an illustration, we can easily encode the two O'Donoghue-Rabin
examples discussed earlier via the two small graphs in 
Figure \ref{fig:intro-ex01}.
In Figure \ref{fig:one-fan}, representing the 2-period example,
the direct edge from $s$ to $t$ corresponds to performing the task
immediately, while moving from $s$ to $v_1$ corresponds to waiting
until tomorrow to do the task.
In Figure \ref{fig:two-fan}, representing the 3-period example,
the direct edge from $s$ to $t$ again corresponds to performing the task
immediately, while moving to $v_1$ corresponds to waiting
until (at least) Friday, and continuing on to $v_2$ corresponds to
waiting until Saturday.
The reader can check that optimal, naive, and sophisticated agents
solve their respective path-planning problems on these two graphs
as they do in the examples.
For example, the sophisticated agent in Figure \ref{fig:one-fan} will
move to $v_1$, since from $s$ it perceives this path as costing $c < b$;
but the sophisticated agent in Figure \ref{fig:two-fan} will reason
that its $v_1$-self will move to $v_2$, for an eventual cost that it
perceives as $c^2 > b$, and so its $s$-self will move directly to $t$.

\xhdr{Overview of Results: Cost Ratios}
We are thus bringing together two formalisms, integrating the notion of
sophisticated behavior in present-biased agents
\cite{odonoghue-now-or-later,odonoghue-choice-procrastination},
into the graph-theoretic model of present-biased planning \cite{ko-ec14}.
This synthesis turns out to uncover a rich structure that wasn't
apparent in the long line of earlier behavioral economics work,
including a range of findings that shed new light on sophisticated behavior.

We first consider the question of worst-case guarantees for 
sophisticated agents.
The 2-period example (rendered in graph-theoretic terms in 
Figure \ref{fig:intro-ex01}) shows that when $c$ approaches $b$ from below,
the sophisticated agent is doing a factor of $b$ worse than optimal.
Let us call this the {\em cost ratio}: the ratio of the biased agent's
cost to the optimal cost.
Can we find examples with cost ratios that are 
worse than in this simple structure?

In fact, we can't: in every graph instance, the cost incurred by
the sophisticated agent is at most a factor of $b$ worse than 
the optimal cost.
This tight bound is independent of the size of the graph;
given that the cost incurred by a naive agent can be exponentially worse
in the size of the graph \cite{ko-ec14}, this is a first striking 
indication of the power of sophistication in general instances.

We then consider the same question for a common generalization
of naivete and sophistication --- the notion of 
{\em partial naivete} introduced in O'Donoghue and Rabin's work
\shortcite{odonoghue-choice-procrastination}.
A partially naive agent is aware of its present bias, but might
be wrong about the value of its parameter $b$; in particular, 
it may believe its bias parameter to be $b' \neq b$, and to
plan based on the assumption that the parameter is $b'$.
Thus, a sophisticated agent has $b' = b$, (because it knows its true bias)
while a naive agent has $b' = 1$ (because it believes that it has no bias).
How does the worst-case performance of a partially naive agent
vary in the value of $b'$?\footnote{O'Donoghue and Rabin
consider only the case of $b' \leq b$, when agents underestimate
their present bias, as naive agents do in an extreme sense.
But the case of $b' > b$ is interesting as well, and forms a 
powerful contrast as we will see.}

We show that there is a surprisingly
sharp transition in performance as we vary $b'$, with 
the fully sophisticated case $b' = b$ 
forming the boundary of the transition.
Specifically, agents who are {\em pessimistic},
believing that they overweight costs by a factor of $b' > b$,
have a worst-case cost ratio of at most $b'$; but
{\em optimistic} agents, who believe that 
they overweight costs by a factor less than $b$,
can perform as badly as naive agents, with an
exponential cost ratio.

\xhdr{Overview of Results: Rewards and Commitment Devices}
We next consider extensions of the basic model that incorporate
rewards in addition to costs, again finding new phenomena
in sophisticated behavior.

First, we consider the case in which the agent is not required
to reach the target node $t$; instead, there is a reward $R$ at $t$, and
the agent collects the reward only if it reaches $t$.
The agent thus faces the question both of {\em whether} to attempt to 
reach $t$, and if so then which path to follow.
Naive agents will sometimes start on a path to $t$ and then
give up partway through the path, but a sophisticated agent
will never do this; because it correctly predicts its future 
decision-making, it is able to determine ahead of time whether
it will be able to get from $s$ to $t$ or not.

Despite this, it turns out to be 
quite complex to understand the set of 
{\em feasible rewards} that will cause a sophisticated
agent to travel from $s$ to $t$.
In particular, we discover a novel Braess-Paradox-like phenomenon: 
there are instances in which the 
agent will decide to traverse the graph for a reward $R$, but will 
decide not to traverse the graph for a certain larger reward $R' > R$.
This surprising but direct consequence of sophistication does not
appear to have analogues in the prior literature on present bias.
The intuition behind this phenomenon
is that increasing the reward can make parts of the graph
traversable to one of the agent's future selves that 
previously were not traversable, and this can lead an earlier self 
to perceive a higher cost.
This non-monotonicity implies that the set of
feasible rewards does not necessarily form a connected set.
In fact, we show how to construct instances for which
the set of feasible rewards consists of exponentially many
disjoint intervals.

We also consider the ``pure reward'' case in which an agent
is seeking to collect rewards, rather than to minimize costs,
throughout its traversal of the graph.
This is easily modeled in our framework by simply having 
the edge weights in our directed acyclic graph $G$ correspond
to rewards rather than costs; as before, the agent seeks a path in $G$
from $s$ to $t$, but the goal is now to traverse as much edge weight
as possible.  O'Donoghue and Rabin 
\shortcite{odonoghue-now-or-later,odonoghue-choice-procrastination},
had observed that there are
cases where sophisticated agents tend to do badly at collecting rewards
in particular
and we find this effect borne out and generalized
to arbitrary instances in the graph-theoretic model:
there are instances on which a sophisticated agent will collect
an exponentially small fraction of the optimal reward.
Interestingly, we also show that naive agents always collect
at least a $1/b$ fraction of the optimal reward, a striking
reversal of the relative worst-case 
performance between sophistication and naivete
when the goal is gaining reward rather than minimizing cost.

Given that sophisticated agents can have very bad worst-case
behavior in this reward-seeking case, it is a natural setting
to model the power of different kinds of {\em commitment devices}
for improving their outcomes.
Commitment devices are familiar from both laboratory settings
and everyday experience as mechanisms to intentionally limit
future options and thus cut down opportunities for sub-optimal 
behavior \cite{brocas2004commitment}.
(A common example is via commitments to physical fitness program
\cite{dellavigna2006paying}).
Given our current model, we can ask how powerful such commitment devices
may be in reducing the effects of present bias with sophistication, and
how they can best be constructed.

In our setting, we can model commitment devices through 
operations that represent a controlled ability to modify the graph
instance at the outset.
We show that such commitment devices have a powerful ability
to reduce the worst-case in reward-seeking instances, 
through results showing roughly that every
instance has a small modification that reduces the exponential worst-case
to something much smaller.
Perhaps most interesting among these is the fact that 
in any reward-seeking instance,
and for any $\eps > 0$, it is possible to remove at most an
$\eps$ fraction of the edges and bring the sophisticated agent's
performance on the modified graph to within a factor of
$\Omega(n^{-\eps/2})$ of optimal --- an exponential
improvement in the general worst-case guarantee compared
to the performance on the unmodified graph.

%% file: cost-ratio.tex

A sophisticated agent plans similarly to a naive agent with present
bias, but the sophisticated agent takes into account the fact that its
future selves will also exhibit present bias. This means that in order
to compute the path that it plans to take, it first computes the path
that its future self will take starting from each subsequent node
in the topological ordering of the underlying directed acyclic graph $G$. 
Formally, we can use the following
recursive formula to define 
the cost $C_s(u)$ of the path taken by a sophisticated agent with present-bias
parameter $b$, starting from a node $u$:
\begin{align*} 
  C_s(u) = c(u, S_s(u)) + C_s(S_s(u)) \text{ where } S_s(u) = \argmin_{v : (u,v)
  \in E } ~ b \cdot c(u, v) + C_s(v).
\end{align*}

Note that the path that the agent will take can be computed in
polynomial time as the graph is a DAG and $C_s(u)$ only depends on
vertices that come after $u$ in the topological ordering of $G$.

To get a better understanding for the complexity of a sophisticated
agent's behavior, it is useful to go over an example
that illustrate some non-obvious consequences of sophistication.
In particular, consider the graph
in Figure~\ref{fig:change}, 
where we will show how
an agent with present bias $b = 1$ or $b = 10$ will
traverse the same path, but an agent with a present bias $b = 2$
strictly between them will traverse a different path.
We note that this is a contrast with naive agents, where
the set of all $b$ for which an agent traverses traverses a given path $P$
forms an interval \cite{ko-ec14}.

To understand where this non-monotonic phenomenon comes from, 
we observe that for $b = 1$, the
agent acts optimally and takes the $s \to u \to t$ path. However, 
for other values of $b$, we have the following:
\begin{itemize}
\item 
When $b = 2$,
the sophisticated agent creates a different plan by reasoning about its
future selves and their perceived costs as follows.
It finds that at $v$, the perceived
cost of going from $v$ to $t$ is $b \cdot 5 = 10$, while the perceived cost of
taking the $v \to w \to t$ path is $b \cdot 0 + 49 = 49$. 
Thus, with $b = 2$, the agent standing at $v$
would go directly from $v$ to $t$. 
The sophisticated agent knows this while standing at $s$, so 
when computing a path from $s$,
it calculates that the perceived cost of the $s \to u \to t$ 
path is $b \cdot 3 + 2 = 8$, while the perceived cost of the 
$s \to v \to t$ path is $b \cdot 1 + 5 = 7$.
This means that it chooses the $s \to v \to t$ path. 
\item 
When $b = 10$, 
the agent does something different --- from $v$, its perceived cost of
taking the $v \to t$ path is $b \cdot 5 = 50$, while its perceived cost of
taking the $v \to w \to t$ path is $b \cdot 0 + 49 = 49$. Thus, if it reached
$v$, it would take the $v \to w \to t$ path. Knowing this, when choosing a path
from $s$, its only options are $s \to u \to t$, 
which has perceived cost $b \cdot 3 + 2 = 32$, 
and $s \to v \to w \to t$, which has perceived cost 
$b \cdot 1 + 49 = 59$. 
Thus, it chooses to follow the $s \to u \to t$ path ---
as the optimal agent did, but for different reasons.
\end{itemize}

    \figureswraper{
    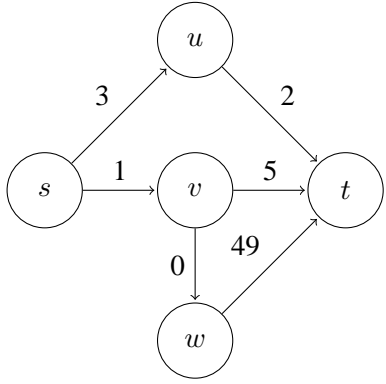
\begin{figure}[t]
      \centering
      \begin{tikzpicture}[->,shorten >=1pt,auto,node distance=2cm, thin]
        \node (0) [circle, draw, minimum size=1cm] at (0,0) {$s$};
        \node (1) [circle, draw, minimum size=1cm] at (2,2) {$u$};
        \node (2) [circle, draw, minimum size=1cm] at (2,0) {$v$};
        \node (3) [circle, draw, minimum size=1cm] at (2,-2) {$w$};
        \node (4) [circle, draw, minimum size=1cm] at (4,0) {$t$};

        \path
        (0) edge node [above left] {3} (1)
        (0) edge node [above] {1} (2)
        (1) edge node [above right] {2} (4)
        (2) edge node [left] {0} (3)
        (2) edge node [above] {5} (4)
        (3) edge node [above left] {49} (4)
        ;
      \end{tikzpicture}
      \caption{Graph $G$ for which the paths for $b = 1$ and $b = 10$ are
        the same, but different from the path for $b = 2$ \label{fig:change}}
    \end{figure}
    } 

\subsection{General Bounds on the Cost Ratio}
A first natural question to ask is to compare the performance of sophisticated
agents relative to optimal ones (who will simply take the shortest path). In the
introduction we have seen an example (Figure \ref{fig:one-fan}) in which the
cost ratio between the cost exhibited by a sophisticated agent and the shortest
path can be $b$. Here we show that this cost ratio is indeed the worst
attainable. This is in contrast to the cost ratio for naive agents,
which in the worst case can be
exponential in the graph size.

\begin{theorem} \label{thm:full}
The cost ratio for a sophisticated agent is at most $b$.
\end{theorem}
\begin{proof}
Denote the cost of the optimal (i.e. min-cost) path from $u$ by $C_o(u)$ and the node that follows $u$ on the optimal path by $S_o(u)$. We use induction on the height of a node, defined as the number of edges in the longest 
unweighted path from the node to $t$.

\textbf{Inductive hypothesis:} For a node $u$ with height $k$, $C_s(u) \le
  b \cdot C_o(u)$.

\textbf{Base case:} For a node $u$ of height 1, $C_s(u) = C_o(u) \le b
\cdot C_o(u)$, as the sophisticated agent will take the optimal
path when there are no paths to the goal longer than one step.
 
\textbf{Inductive step:} For a node $u$ with height $k+1$, let $v_s = S_s(u)$
    and let $v_o = S_o(u)$. By definition, if the sophisticated agent chose to go to $v_s$ from $u$ instead of $v_o$, that means that $b \cdot c(u,v_s) + C_s(v_s) \le b \cdot c(u,v_o) + C_s(v_o)$. Since $b > 1$ and by definition $C_s(u) = c(u,v_s) + C_s(v_s)$, this implies that $C_s(u) \le b \cdot c(u,v_o) + C_s(v_o)$. Recall that by the induction hypothesis we have that $C_s(v_o) \le b \cdot C_o(v_o)$. Thus we have that $C_s(u) \le b \cdot c(u,v_o) + b \cdot C_o(v_o) = b \cdot C_o(u)$.
\end{proof}
\input{naive-vs-soph}

\yhdr{Partially Naive Agents}
Next, we observe that quantifying the cost ratio is interesting not
only for sophisticated agents but also for agents that are \emph{partially
naive}. Following O'Donoghue and Rabin
\cite{odonoghue-choice-procrastination}, these agents are characterized
by two parameters: they have a true present-bias parameter of $b$,
but they act based on a belief that their present-bias parameter is 
$b'$.
For sophisticated agents we have that $b'=b$, while for partially naive
agents we have $b' \neq b$. DellaVigna \cite{dellavigna2007psychology} 
summarizes several lines 
of evidence that most people are partially naive, in the sense that
they are optimistic about their present bias ($b'<b$). In Appendix
\ref{app:partial-naive} we show that, in fact, 
optimistic agents do not outperform naive agents in the worst case,
and they can have exponential cost ratio.
The situation is very different for 
partially naive agents who are pessimistic
($b'>b$); their cost ratio is bounded by $b'$.

Thus, there is an exponential transition in the worst-case cost ratio
as we vary the partial naivete parameter $b'$, with the 
case of sophistication $b' = b$ forming the boundary of the transition.

%% file: naive-vs-soph.tex
In the introduction we already observed that there are instances in which a
sophisticated agent is doing significantly better than a naive agent, and now we
have a strong worst-case guarantee for the cost ratio for a sophisticated agent.
Given this, one might suspect that sophisticated agents are always better off
than naive ones in every instance.
In fact, this is not the case; we now
show that there are instances in which a naive agent
can do better than a sophisticated one by a factor arbitrarily close to $b$.
(Note that that by Theorem~\ref{thm:full}, it cannot exceed $b$.)

\begin{claim}
There are instances on which the cost incurred by a sophisticated agent
can exceed the cost incurred by a naive agent by a factor arbitrarily
close to $b$.
\end{claim}
\begin{proof}
Figure \ref{fig:sophnaive} shows how to construct such an instance.
In the figure, we will think of $x$ as very large relative to $b$,
and $\eps$ as very close to $0$; at the end, we will let these parameters
go toward their limits.

At the
beginning, the naive agent perceives that the cost of the path $s \to u \to t$
is $bx + 1$ and the perceived cost of the path $s \to w \to t$ is $bx + b -
2\varepsilon$, so it chooses to go to $u$. Once at $u$, it chooses the path $u
\to v \to t$ instead, for a total incurred cost of $x + b - \varepsilon$. The
sophisticated agent, however, understands that once it reaches $u$, it will
proceed to $v$ instead of $t$. Therefore, its perceived cost of a path going
through $u$ is $bx + b - \varepsilon$, while the perceived cost of going through
$w$ is $bx + b - 2\varepsilon$. Therefore, it takes the path $s \to w \to t$ for
a total cost of $bx + b - 2\varepsilon$. The ratio of the sophisticated
agent's cost to the naive agent's cost is therefore $(bx + b -
2\varepsilon) / (x + b - \varepsilon)$. Hence, in the limit as $x$ approaches
infinity this ratio approaches $b$.
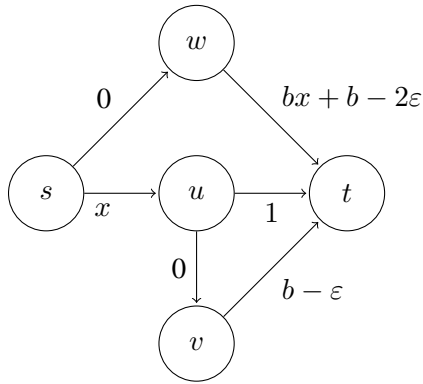
\begin{figure}[ht]
  \centering
  \begin{tikzpicture}[->,shorten >=1pt,auto,node distance=2cm, thin]
    \node (s) [circle, draw, minimum size=1cm] at (0,0) {$s$};
    \node (u) [circle, draw, minimum size=1cm] at (2,0) {$u$};
    \node (v) [circle, draw, minimum size=1cm] at (2,-2) {$v$};
    \node (w) [circle, draw, minimum size=1cm] at (2,2) {$w$};
    \node (t) [circle, draw, minimum size=1cm] at (4,0) {$t$};

    \path
    (s) edge node [below left] {$x$} (u)
    (s) edge node {0} (w)
    (u) edge node [below] {$1$} (t)
    (u) edge node [left] {0} (v)
    (v) edge node [below right] {$b - \varepsilon$} (t)
    (w) edge node {$bx + b - 2\varepsilon$} (t)
    ;
  \end{tikzpicture}
  \caption{Worst-case ratio for the cost of a sophisticated agent relative
to a naive agent}
  \label{fig:sophnaive}
\end{figure}
\end{proof}

%% file: reward-at-target.tex
We now adapt the model to consider cases in which a sophisticated agent
must decide
whether or not the task is worth completing.
We place a reward $R$ at the goal vertex $t$, and
the agent will only choose to continue towards the goal if the reward at $t$
is at least as large as the perceived cost of reaching $t$. 
At each node, the agent
chooses its path as before, but it also has the ability to
{\em abandon} the traversal toward $t$ at any point
if the perceived cost is higher than the reward.
Furthermore, the agent engages in sophisticated reasoning as before,
and so it can determine for each future node $v$ in $G$ whether
or not it would abandon its traversal if it were at $v$,
and it will avoid plans that take it to nodes at which it would
abandon the traversal.

We can perform this reasoning to identify the nodes where the traversal
of $G$ will be abandoned, working backward over a topological ordering of $G$.
We give the details via an algorithm in the proof of the following claim.

\begin{claim}
Given a graph $G$ and a reward $R$ placed on $t$, we can determine in linear time whether or not the agent will reach $t$, and if so, which path it will take.
\label{claim:reach-t-fixed-reward}
\end{claim}

\begin{proof}
Finding the path an agent takes
through $G$ can be done by a two-phase process.  In the first phase,
we iteratively determine, for each node $v$, whether or not the agent
would abandon the traversal starting from $v$.
We prune all the nodes for which the agent would abandon the traversal.
Second, in
the pruned graph, $G'$, the agent will take the same path as agent
would take in $G'$ without considering a reward. As any vertex/edge
not pruned by this process can be freely used by the agent, the agent will
choose between them as it does in the absence of a reward.

To prune the graph, we iterate over the vertices
in reverse topological order. For each vertex $v$ and each edge $e$
originating from it we can compute the perceived cost of a path that
the agent will take; since the subgraph starting from $v$ is already
pruned this path will be simply the path that an agent will traverse
without rewards. If this perceived cost is greater than $R$ then the
edge is pruned. If all the outgoing edges of a vertex are pruned, we
prune the vertex as well.
\end{proof}

In this section we consider two types of questions regarding the model
with a reward at the target node $t$.
In Section \ref{sec-traversable} we ask, given a
graph $G$, which rewards will motivate the agent to reach $t$,
and how different rewards change both the agent's decision to traverse
the graph and the path the agent follows.
In particular, we show that the minimum reward required for motivating
a sophisticated agent to reach $t$ is 
at most $b$ times the minimum cost of a path in the absence of a reward;
and moreover, the number of different paths 
the agent takes as we vary the reward $R$ can be exponential. 
In Section
\ref{sec-modifying} we consider further modifications to the graph 
as a means of motivating the agent to reach $t$.
Specifically, we consider two methods for
modifying the graph: (i) deleting edges, where we show that a minimal
subgraph motivating the agent to reach $t$ must be
a path; and (ii) adding internal rewards on the
edges, which we show is less powerful than edge deletion by a factor of at
least $2$. We conclude this section by some open questions about rewards.

\subsection{Traversable Rewards} \label{sec-traversable}

\def\Rmin{R^{\rm min}}

For the remainder of this section, we let $\Rmin$ denote the
minimum reward needed to motivate a sophisticated agent to reach $t$
in the graph $G$.
Before turning to our results, we note the example
in Figure \ref{fig:followp}, which shows 
that the reward $\Rmin$ might motivate an agent to take 
a path other than the optimal one. The optimal path, $s \to u \to t$, would
require a reward of 6 to motivate the agent to follow it in isolation. However,
the full graph is not traversable for a reward of 6, since at $u$, the agent
would prefer to follow the $u \to v \to t$ path, which has perceived cost 3.
Note that the agent would not continue from $w$ for a reward of 6, as the
perceived cost of going from $w$ to $t$ is 10. For a reward of 7, however, the
agent would be willing to follow the $s \to u \to v \to t$ path, and thus,
$\Rmin = 7$. For a reward of 10, the agent is now willing to go from $w$ to $t$.
From $v$, it would choose the $v \to w \to t$ path over the $v \to t$ path since
it has a lower perceived cost. As a result, from $u$, the perceived cost of
going to $v$ is 5 (because it knows that from $v$ it would proceed to $w$),
while the perceived cost of going direclty to $t$ is $b \cdot 2 = 4$. This means
that the agent follows the $s \to u \to t$ path.

\figureswraper{
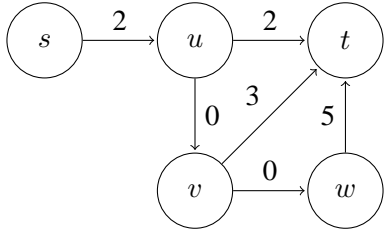
\begin{figure}[t]
  \centering
  \begin{tikzpicture}[->,shorten >=1pt,auto,node distance=2cm, thin]
    \node (s) [circle, draw, minimum size=1cm] at (-2,0) {$s$};
    \node (u) [circle, draw, minimum size=1cm] at (0,0) {$u$};
    \node (v) [circle, draw, minimum size=1cm] at (0,-2) {$v$};
    \node (w) [circle, draw, minimum size=1cm] at (2,-2) {$w$};
    \node (t) [circle, draw, minimum size=1cm] at (2,0) {$t$};

    \path
    (s) edge node {2} (u)
    (u) edge node {2} (t)
    (u) edge node {0} (v)
    (v) edge node {3} (t)
    (v) edge node {0} (w)
    (w) edge node {5} (t)
    ;
  \end{tikzpicture}
  \caption{$\Rmin$ may be larger than the minimum reward required to traverse
    any path in isolation. \label{fig:followp}}
\end{figure}
}

Even though the minimum reward required to motivate the agent to
traverse the graph can be greater than the reward required to motivate
the agent to follow the optimal path in isolation, we show that this
minimum reward is at most $b$ times the cost of the optimal path. 
Specifically, we prove the following result \EC{in 
Appendix~\ref{app:rewardratio}}:

\begin{theorem} \label{thm:rewardratio}
  Let $R$ be the minimal reward required to motivate a sophisticated agent to
  reach $t$. Then, $R \le b \cdot C_o(s)$
\end{theorem}

\arxiv{
\input{ub-traversable}
}

\EC{
It is easy to see that this bound is tight by considering a graph which is
only a single edge with a cost $c$. For such a graph $R=b \cdot c$ and $C_o(s) =
c$. Furthermore, let $R_d$ be the reward required to traverse any path in $G$ in
isolation. Since $R_d \ge C_o(s)$, we have that $R \le b \cdot R_d$.
Figure~\ref{fig:rtight} in Appendix~\ref{app:rewardratio} shows that this bound
is tight as well.
}

We know that both $R$ and $C_s(s)$ are bounded by $b \cdot C_o(s)$ by
Theorems~\ref{thm:full} and~\ref{thm:rewardratio}. However, as shown in
Figure~\ref{fig:min_inf_reward}, neither is necessarily bounded by the other.
With $b = 2$, in $G_1$, $R = 2$ and $C_s(s) = 1$. In $G_2$, $R = 11$ and $C_s(s)
= 12$.
\begin{figure}[ht]
  \centering
  \subfloat[In $G_1$, $R > C_s(s)$]{
    \centering
    \begin{tikzpicture}[->,shorten >=1pt,auto,node distance=2cm, thin]
      \node (s) [circle, draw, minimum size=1cm] at (0,0) {$s$};
      \node (t) [circle, draw, minimum size=1cm] at (2,0) {$t$};
      
      \path
      (s) edge node {1} (t)
      ;
    \end{tikzpicture}
  }
  ~ 
  \subfloat[In $G_2$, $C_s(s) > R$]{
    \centering
    \begin{tikzpicture}[->,shorten >=1pt,auto,node distance=2cm, thin]
      \node (s) [circle, draw, minimum size=1cm] at (-2,0) {$s$};
      \node (u) [circle, draw, minimum size=1cm] at (0,0) {$u$};
      \node (v) [circle, draw, minimum size=1cm] at (2,0) {$v$};
      \node (w) [circle, draw, minimum size=1cm] at (3,-1.5) {$w$};
      \node (t) [circle, draw, minimum size=1cm] at (4,0) {$t$};

      \path
      (s) edge node {$2$} (u)
      (u) edge node {$4$} (v)
      (v) edge node {0} (w)
      (v) edge node {3} (t)
      (w) edge node {6} (t)
      ;
    \end{tikzpicture}
  }
  \caption{There is no strict relation between $R$ and $C_s(s)$. These examples
  use $b=2$.}
  \label{fig:min_inf_reward}
\end{figure}

%

\yhdr{Path-counting over Different Rewards}
As discussed in the introduction, it is a surprising fact
that the agent's decision to traverse
the graph need not be monotone in the reward $R$: it is possible for the
agent to decide to traverse the graph for a reward $R$, but to decide
not to for a larger reward $R'$.
We begin by illustrating this in a simple example:
\begin{example} 
Consider the graph in Figure~\ref{fig:non-monotone} with
$b = 2$. We observe that $G$ is traversable for 
$R_1 = 9$ and $R_3 = 11$, but is not traversable for 
$R_2 = 10$. This is because for a reward of 9, $w$ is abandoned, as the
perceived cost of going from $w$ to $t$ is 10. This means that the only
possible path is $s \to v \to t$, and the perceived costs at $s$ and $v$ are 9
and 6 respectively. For a reward of 10, $w$ is no longer abandoned, meaning from
$v$, the agent would choose to go to $w$ instead of $t$. As a result, the
perceived cost from $s$ is 11, and the agent would be unwilling to start. For a
reward of 11, the agent is willing to follow the $s \to v \to w \to t$ path.

\figureswraper{
\begin{figure}[ht]
  \centering
  \begin{tikzpicture}[->,shorten >=1pt,auto,node distance=2cm, thin]
    \node (s) [circle, draw, minimum size=1cm] at (0,0) {$s$};
    \node (v) [circle, draw, minimum size=1cm] at (2,0) {$v$};
    \node (w) [circle, draw, minimum size=1cm] at (3,-1.5) {$w$};
    \node (t) [circle, draw, minimum size=1cm] at (4,0) {$t$};

    \path
    (s) edge node {3} (v)
    (v) edge node {0} (w)
    (v) edge node {3} (t)
    (w) edge node {5} (t)
    ;
  \end{tikzpicture}
  \caption{For $b = 2$, the graph is traversable for $R_1 = 9$ and $R_3 = 11$ but not $R_2 =
  10$} \label{fig:non-monotone}
\end{figure}
} 
\end{example}

Next we show that an agents can traverse exponentially many different path for different rewards.
\begin{theorem}
There exists a graph $G$, and a set $S$ consisting of an 
exponential number of different rewards, 
such that for every reward in $S$ the agent will traverse a different path.
\end{theorem}
\begin{proof}
The graph shown in Figure~\ref{fig:exp_paths} functions like a binary counter --
each $v_i \to v_{i+1}$ structure represents a bit, where $i=0$ corresponds to
the least significant bit. If the agent goes directly from $v_i$ to $v_{i+1}$,
that denotes a binary 0, while following the $v_i \to w_i \to v_{i+1}$ path
denotes a binary 1. With the proper choice of $b$ and $c$, we can get the agent
to follow a path corresponding to each number between 0 and $2^n - 1$.

\figureswraper{
\begin{figure}[ht]
  \centering
  \begin{tikzpicture}[->,shorten >=1pt,auto,node distance=2cm, thin]
    \node (s) [circle, draw, minimum size=1cm] at (-6,0) {$s$};
    \node (v0) [circle, draw, minimum size=1cm] at (-4,0) {$v_0$};
    \node (v1) [circle, draw, minimum size=1cm] at (-2,0) {$v_1$};
    \node (v2) [circle, draw, minimum size=1cm] at (0,0) {$v_2$};
    \node (d) at (2, 0) {$\cdots$};
    \node (v3) [circle, draw, minimum size=1cm] at (3,0) {$v_{n-1}$};
    \node (t) [circle, draw, minimum size=1cm] at (5,0) {$t$};
    \node [below = 1cm of v0] (w0) [circle, draw, minimum size=1cm] {$w_0$};
    \node [below = 1cm of v1] (w1) [circle, draw, minimum size=1cm] {$w_1$};
    \node [below = 1cm of v2] (w2) [circle, draw, minimum size=1cm] {$w_2$};
    \node [below = 1cm of v3] (w3) [circle, draw, minimum size=1cm] {$w_{n-1}$};

    \path
    (v0) edge [decorate] node [above] {1} (v1)
    (v1) edge [decorate] node [above] {2} (v2)
    (v2) edge [decorate] node [above] {4} (d)
    (v3) edge [decorate] node [above] {$2^{n-1}$} (t)

    (s) edge node {0} (v0)

    (v0) edge node {0} (w0)
    (v1) edge node {0} (w1)
    (v2) edge node {0} (w2)
    (v3) edge node {0} (w3)

    (w0) edge node [below right] {$c$} (v1)
    (w1) edge node [below right] {$2c$} (v2)
    (w2) edge node [below right] {$4c$} (d)
    (w3) edge node [below right] {$c \cdot 2^{n-1}$} (t)
    ;
  \end{tikzpicture}
  \caption{Graph with exponentially many paths}
  \label{fig:exp_paths}
\end{figure}
} 

Specifically, for a reward of $2^n + (bc-2)x$, the agent will follow the path
corresponding to the binary representation of $x$ (where $x \le 2^n - 1$), under
the following conditions: (1) $c > 1$, (2) $b \ge c$ and (3) $1 + c = bc$. Note that the last condition implies that $b < 2$ because otherwise, we would
have $c \le 1$. Furthermore, the first and third conditions imply $bc > 2$
because $1 + c > 2$. We defer the proof to Appendix~\ref{app:exp-paths}.
\end{proof}

Moreover, by considering $w_0$ to be the start vertex of this
instance, we obtain an example in which the set of rewards for which
the agent will traverse the graph consists of exponentially many
disjoint intervals, since $w_0$ is abandoned for reward $2^n + (bc-2)x$ with
even values of $x$ between 0 and $2^n-1$.

\subsection{Modifying the Graph}  \label{sec-modifying}
We now consider two methods of modifying the graph to 
reduce the reward required
to make it traversable: (i) deleting edges from the graph,
and (ii) placing rewards on internal edges instead of just on $t$.

\yhdr{Motivating Subgraphs and Edge Deletion}

If we have the ability to delete edges from $G$, with a reward $R$
at $t$, we can ask about
the structure of a {\em minimal motivating subgraph}: a subgraph of $G$
with the property that the agent will traverse it, but will not traverse
any proper subgraph of it.
For sophisticated agents a minimal motivating subgraph is always a single path. Intuitively, we observe that if an agent is willing to traverse a path $P$ as part of a larger graph $G$ for a reward $R$, then it is willing to
traverse $P$ in isolation for $R$.  \EC{In Appendix~\ref{app:deletion}} \arxiv{Next,} we prove this and show an algorithm for computing the minimal reward required to motivate the agent to reach $t$ in any subgraph.
\begin{claim} \label{clm:single-path}
  A minimal motivating subgraph of a graph $G$ must be a single path.
\end{claim}
\arxiv{
\input{edge-deletion}
}
\EC{
\begin{claim} \label{clm:reward-alg}
  In polynomial time we can find the minimum $R$ such that
  placing a reward of $R$ at $t$ and deleting a subset of edges from $G$
  will motivate the agent to reach $t$ in the resulting subgraph.
\end{claim}
}

\yhdr{Adding Internal Rewards to the Graph}
We also consider a slightly less powerful model, in which instead of deleting
edges, we place rewards on internal edges instead of only on $t$. More
precisely, the agent collects some reward $r(u,v)$ after traversing the edge
$(u,v)$. Note that the reward is only collected after the current time step, so
its value does not get multiplied by $b$.

Let $R_i$ be the minimum reward that can be distributed among the
edges to motivate the agent to reach $t$ and let $R_d$ be the minimal
reward required to motivate the agent to reach $t$ in any subgraph of $G$ (i.e.
enabling edge deletion). We first observe that internal rewards
cannot be more powerful than arbitrary deletion (i.e., $R_i \geq
R_d$). To see why, let $P$ be the path that the agent takes in the
graph with internal rewards $R_i$. It has to be the case that if the
agent traverses $P$ in the graph it will also traverse it in isolation
(as established in Claim \ref{clm:single-path}). Now, for a path in
isolation it is easy to see that the best way to distribute rewards is
put all of them at $t$. Thus, the agent will traverse $P$ for
a reward of $R_i$ on $t$.

Moreover, edge deletion is strictly more powerful than internal rewards as we illustrate next.

\begin{example}
The graph in Figure~\ref{fig:rewarddist} with $b=4$ and $R=52$. $G'$
is traversable under these parameters, but $G$ is not. Furthermore, it is
impossible to distribute $R$ along the edges of $G$ to make it traversable.

\begin{figure}[ht]
  \centering
  \subfloat[$G$]{
    \centering
    \begin{tikzpicture}[->,shorten >=1pt,auto,node distance=2cm, thin]
      \node (s) [circle, draw, minimum size=1cm] at (-2,0) {$s$};
      \node (u) [circle, draw, minimum size=1cm] at (0,0) {$u$};
      \node (v) [circle, draw, minimum size=1cm] at (1,-1.5) {$v$};
      \node (w) [circle, draw, minimum size=1cm] at (2,0) {$w$};
      \node (t) [circle, draw, minimum size=1cm] at (4,0) {$t$};
      
      \path
      (s) edge node {8} (u)
      (u) edge node {10} (w)
      (u) edge node [below left] {5} (v)
      (v) edge node [below right] {6} (w)
      (w) edge node {10} (t)
      ;
    \end{tikzpicture}
  }
  ~ 
  \subfloat[$G'$]{
    \centering
    \begin{tikzpicture}[->,shorten >=1pt,auto,node distance=2cm, thin]
      \node (s) [circle, draw, minimum size=1cm] at (-2,0) {$s$};
      \node (u) [circle, draw, minimum size=1cm] at (0,0) {$u$};
      \node (v) [circle, minimum size=1cm] at (1,-1.5) {};
      \node (w) [circle, draw, minimum size=1cm] at (2,0) {$w$};
      \node (t) [circle, draw, minimum size=1cm] at (4,0) {$t$};
      
      \path
      (s) edge node {8} (u)
      (u) edge node {10} (w)
      (w) edge node {10} (t)
      ;
    \end{tikzpicture}
  }
  \caption{With $b=4$ and $R=52$, it is impossible to distribute $R$ such that
  the agent follows the $s \to u \to w \to t$ path}
  \label{fig:rewarddist}
\end{figure}
\end{example}

With a more complex construction, we show the following
in Appendix~\ref{app:internal-ratio}.
\begin{theorem}
There exist instances in which the ratio between $R_d$ and $R_i$ is as high as $2-\eps$ for any $\eps>0$.
\end{theorem}

Combining the facts in this section with the result of 
Theorem \ref{thm:rewardratio}, we have the following chain 
of inequalities that summarizes the relationships among the
quantities we have considered:
$C_o(s) \leq R_d \leq R_i \leq \Rmin \leq b C_o(s)$.

\subsection{Two Open Questions}
Finally, we note two interesting open questions based on the results
of this section.

First, because we can determine for any $R$ whether the agent
will traverse for a reward of $R$ at $t$, it follows that there is
a simple pseudo-polynomial algorithm to find the minimum $R^*$ at $t$
for which $G$ is traversable: we simple try each value of $R$ up to
$b C_o(s)$ and take the minimum.
However, because the set of $R$ for which $G$ is traversable need
not form a connected set, there is no natural way to use the
decision procedure for a single $R$ to perform binary for the minimum $R$.
Indeed, it is an intriguing open question whether the minimum reward $R^*$
can be found in polynomial time.

Second, as noted above, we do not know how large the ratio of $R_i$
to $R_d$ can be.

%% file: ub-traversable.tex
\begin{proof}
  First, we observe that finding a path when there is a reward is equivalent to
  removing all vertices and edges from $G$ that would be abandoned by the
  agent, and then finding a path in this subgraph $G'$. This is because the
  agent ignores any vertices or edges that would lead to abandonment, but it
  still follows the same formula of choosing the successor vertex that minimizes
  the perceived cost.

  We proceed by induction on the height of the vertex, proving that at any
  vertex $u$, if we place a reward of at least $b \cdot C_o(u)$ on $t$, $u$ will
  never be abandoned.

  \textbf{Base case:} For a vertex of height 1 with some cost $x$ to get to $t$,
  placing a reward of at least $bx$ will always motivate the agent to reach $t$,
  and thus the vertex will never be abandoned.
  
  \textbf{Inductive hypothesis:} For a vertex $v$ of height $k$, a reward of $b
  \cdot C_o(v)$ placed on $t$ will never lead to $v$ being abandoned.

  \textbf{Inductive step:} Let $v$ be the next vertex on the optimal path from
  $u$ to $t$, where $u$ has height $k+1$ and $v$ has height $\le k$. Consider
  placing a reward $R_u \ge b \cdot C_o(u)$ on $t$. For such a reward, by
  induction, $v$ will not be abandoned. Let $x = c(u,v)$. We know that
  \begin{align*}
    R_u &\ge b \cdot C_o(u) \\
        &= bx + b \cdot C_o(v) \\
        &\ge bx + C_s(v) \tag{by Theorem~\ref{thm:full}}
  \end{align*}
  The last step holds because by induction, no vertex on the optimal path from
  $v$ to $t$ would be abandoned for a reward of $R_u$. Choosing the path the
  agent would take from $v$ is equivalent to pruning the abandoned nodes from
  $G$ and then choosing a path in the pruned $G'$ as if there is no reward.
  Since $G'$ contains the optimal path from $v$ to $t$ and $C_s(v)$ is the cost
  of the path taken by the agent in $G'$, we can use Theorem~\ref{thm:full} to get
  $C_s(v) \le b \cdot C_o(v)$.

  Since $R_u \ge bx + C_s(v)$, the reward is more than the perceived cost of
  going from $u$ to $t$ through $v$, meaning there is at least one edge that the
  agent is willing to take leading out of $u$. Thus, $u$ will never be abandoned
  for $R_u \ge b \cdot C_o(u)$.

  By induction, we see that $s$ is never abandoned for a reward of at least $b
  \cdot C_o(s)$, meaning that for such a reward, there is a valid path from $s$
  to $t$ in the pruned $G'$. Therefore, for any reward $b \cdot C_o(s)$ or
  higher, $G$ is traversable, meaning the minimum reward for which $G$ is
  traversable is at most $b \cdot C_o(s)$.
\end{proof}

\arxiv{
It is easy to see that the bound in Theorem \ref{thm:rewardratio} is tight by considering a graph which is
only a single edge with a cost $c$. For such a graph $R=b \cdot c$ and $C_o(s) =
c$. Furthermore, let $R_d$ be the reward required to traverse any path in $G$ in
isolation. Since $R_d \ge C_o(s)$, we have that $R \le b \cdot R_d$.
Figure~\ref{fig:rtight} shows that this bound is tight as well.
}

\begin{figure}[ht]
  \centering
  \begin{tikzpicture}[->,shorten >=1pt,auto,node distance=2cm, thin]
    \node (s) [circle, draw, minimum size=1cm] at (0,0) {$s$};
    \node (v1) [circle, draw, minimum size=1cm] at (2,0) {$v_1$};
    \node (v2) [circle, draw, minimum size=1cm] at (4,0) {$v_2$};
    \node (d) at (5,0) {$\cdots$};
    \node (vn1) [circle, draw, minimum size=1cm] at (6,0) {$v_{n-1}$};
    \node (vn) [circle, draw, minimum size=1cm] at (8,0) {$v_n$};
    \node (t) [circle, draw, minimum size=1cm] at (10,0) {$t$};

    \path
    (s) edge node {1} (v1)
    (v1) edge node {1} (v2)
    (vn1) edge node {1} (vn)
    (vn) edge node {1} (t)
    ;

    \path [->,decoration=snake]
    (s) edge [decorate, bend right=60] node [below] {$b(1-\varepsilon)$} (v1)
    (v1) edge [decorate,bend right=60] node [below] {$b(1-\varepsilon)$} (v2)
    (vn1) edge [decorate,bend right=60] node [below] {$b(1-\varepsilon)$} (vn)
    (vn) edge [decorate,bend right=60] node [below] {$b(1-\varepsilon)$} (t)
    ;
  \end{tikzpicture}
  \caption{$G$ such that $R^* = b+n$, $R = b(n+1)(1-\varepsilon)$}
  \label{fig:rtight}
\end{figure}

\begin{figure}[ht]
  \centering
  \begin{tikzpicture}[->,shorten >=1pt,auto,node distance=2cm, thin]
    \node (u1) [circle, draw, minimum size=1cm] at (-5,0) {$u$};
    \node (v1) [circle, draw, minimum size=1cm] at (-3,0) {$v$};
    
    \path [->,decoration=snake]
    (u1) edge [decorate,bend right=60] node [below] {$b(1-\varepsilon)$} (v1)
    ;
    \node (u) [circle, draw, minimum size=1cm] at (0,0) {$u$};
    \node (1) [circle, draw, minimum size=1cm] at (2,0) {};
    \node (d) [circle, minimum size=1cm] at (4,0) {$\cdots$};
    \node (3) [circle, draw, minimum size=1cm] at (6,0) {};
    \node (v) [circle, draw, minimum size=1cm] at (8,0) {$v$};
    \draw
    [-,thick,decorate,decoration={brace,amplitude=10pt}](0,.5) -- (8,.5)
    node[black,midway,yshift=0.4cm]
    {$\frac{b(1-\varepsilon)}{\varepsilon}$};
    
    \path
    (u) edge node {$\varepsilon$} (1)
    (1) edge node {$\varepsilon$} (d)
    (d) edge node {$\varepsilon$} (3)
    (3) edge node {$\varepsilon$} (v)
    ;
  \end{tikzpicture}
  \caption{The curved edge represents $\frac{b(1-\varepsilon)}{\varepsilon}$
  edges of cost $\varepsilon$}
  \label{fig:snakeedge}
\end{figure}

In the graph shown in Figure~\ref{fig:rtight}, the reward required for the path
along the edges of cost 1 is $R^* = b+n$, but the minimum reward for which this
graph is traversable is $R = b(n+1)(1-\varepsilon)$. The curved edges represent
$\frac{b(1-\varepsilon)}{\varepsilon}$ edges, each with a cost of $\varepsilon$,
as described in Figure~\ref{fig:snakeedge}. This makes the perceived cost and
the actual cost of traversing the curved path approximately equal as
$\varepsilon$ goes to 0. As $n$ increases, the ratio between $R$ and $R^*$
approaches $b$.

%% file: edge-deletion.tex
\begin{proof}
  For a given path $P$, the minimum reward required to motivate an agent to
  traverse that path in isolation is just the maximum over all vertices of the
  perceived cost of that vertex. At a vertex $u$, the perceived cost is $b \cdot
  c(u,v)$ plus the true cost of the path from $v$ to $t$ where $v$ is the next
  vertex along $P$. As long as the agent follows $P$, there is no graph of which
  $P$ is a strict subgraph that requires a smaller reward to be traversable. To
  see this, let $R$ be the minimum required to motivate the agent to traverse
  $P$ in isolation. This means there is some vertex $u$ such that the perceived
  cost of the portion of $P$ from $u$ to $t$ is $R$. If a reward $R' < R$ is
  placed on $t$, then the perceived cost of following $P$ from $u$ will be
  larger than the reward, and therefore the agent will not follow $P$ from $u$
  to $t$. This means that for any path $P$, if the agent is unwilling to
  traverse $P$ in isolation for a given reward, then it is unwilling to follow
  $P$ in any graph of which $P$ is a strict subgraph. Thus, if the agent is
  willing to follow $P$ for a given reward in a larger graph $G$, then it is
  willing to follow $P$ in isolation, meaning that $P$ is a minimal motivating
  subgraph of $G$.
\end{proof}

\EC{
Next, we prove Claim \ref{clm:reward-alg}
}
\arxiv{
\begin{claim} 
  In polynomial time we can find the minimum $R$ such that
  placing a reward of $R$ at $t$ and deleting a subset of edges from $G$
  will motivate the agent to reach $t$ in the resulting subgraph.
\end{claim} }
\begin{proof}
Given a graph $G$ and a reward $R$, we can find a motivating subgraph
traversable for with $R$ placed on $t$ (if one exists) in linear time. To do so,
we consider each vertex in reverse topological order. At each vertex $u$, we
choose the neighbor $v$ such that $v$ is not abandoned, the perceived cost of
going to $v$ from $u$ is no more than $R$, and the true cost of going through
$v$ is less than it is through any other neighbor that satisfies the first two
conditions. If no such neighbor exists, then $u$ is abandoned. Since we process
the vertices in reverse topological ordering, when we consider $u$, we already
know the true cost of the path through $v$. Finally, we build a path by starting
at $s$ and following the choices of each vertex until we reach $t$, and that
path is a motivating subgraph for $R$. If $s$ is abandoned, there is no path
such that the perceived cost at every point on that path is no more than $R$. If
such a path existed, then no vertex on that path would be abandoned by the above
algorithm (as the next vertex on the path would always satisfy the conditions),
and so $s$ would not be abandoned. Note that this only holds because at each
vertex we minimize the true cost over all possible successors. This ensures that
the perceived cost at any subsequent vertex we consider is minimized as well, so
if there is any way for the perceived cost to be smaller than $R$ along the
entire path, this algorithm will find it.

Let $R_d$ be the minimum reward required to motivate the agent to reach $t$ from
$s$ in any subgraph of $G$. Equivalently, $R_d$ is the minimum reward required
to motivate the agent to reach $t$ given that we can delete arbitrarily many
edges from the graph. Then, the above algorithm can find $R_d$ in polynomial
time simply by binary searching over all rewards between $C_o(s)$ and $b \cdot
C_o(s)$ using the algorithm described above.
\end{proof}

%% file: rewards.tex
So far, we have focused on present-biased agents who seek to minimize the cost
they incur. However, we could also consider a model that consists
only of rewards ---
each edge $(u,v)$ has a reward $r(u,v)$, and the agent seeks to maximize the
reward that it collects while traversing the graph. In this model, 
the value $R_s(u)$ of the path that a sophisticated agent takes 
starting from $u$  
can be defined as follows: 
$$R_s(u) = r(u, S_s^{(R)}(u)) + R_s(S_s^{(R)}(u)) \text{ where } S_s^{(R)}(u) =
\argmax_{v : (u,v) \in E} ~ b \cdot r(u, v) + R_s(v).$$

In this model, the optimal reward is the weight of the heaviest path
from $v$ to $t$ (denoted by $R_o(v)$). With these definitions, we can
again consider the ratio of rewards collected by the optimal agent
to that collected by a sophisticated agent; we refer to this as the 
{\em reward ratio}.
O'Donoghue and Rabin \cite{odonoghue-now-or-later} have observed that
in reward settings sophisticated agents can do quite badly as they
tend to collect rewards much earlier than they should.
In the following example we show that this principle holds in a strong
sense for our graph-theoretic model:
the reward ratio for sophisticated agents can be exponential in the
size of the graph.

\begin{example}
Consider the fan graph in
Figure~\ref{fig:soph_reward_ratio} with the value of
$c$ set slightly below $b$; 
more generally, such a graph could
have $n$ nodes $v_i$ with an edge $(v_i,t)$ of reward $c^i$.
By reverse induction from $v_n$, a sophisticated agent can work out 
that from $v_i$ it would prefer to go directly to $t$ (obtaining
a perceived reward of $bc^i$) than to continue on to $v_{i+1}$
(where, by induction, it would go directly to $t$, for a perceived
reward of $c^{i+1})$.
The consequence of this induction is that the agent will go 
directly from $s$ to $t$, obtaining a reward of $1$.
Intuitively, the agent is smart enough to predict that each of its
future selves will go directly to $t$, and so it would rather
go directly to $t$ immediately.
\figureswraper{
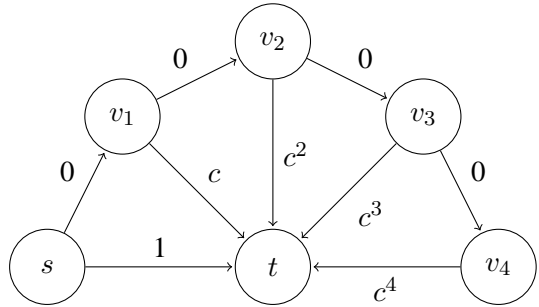
\begin{figure}[ht]
  \centering
  \begin{tikzpicture}[->,shorten >=1pt,auto,node distance=2cm, thin]
    \node (0) [circle, draw, minimum size=1cm] at (0,0) {$s$};
    \node (1) [circle, draw, minimum size=1cm] at (1,2) {$v_1$};
    \node (2) [circle, draw, minimum size=1cm] at (3,3) {$v_2$};
    \node (3) [circle, draw, minimum size=1cm] at (5,2) {$v_3$};
    \node (4) [circle, draw, minimum size=1cm] at (6,0) {$v_4$};
    \node (5) [circle, draw, minimum size=1cm] at (3,0) {$t$};
    
    \path
    (0) edge node {0} (1)
    (1) edge node {0} (2)
    (2) edge node {0} (3)
    (3) edge node {0} (4)

    (0) edge node {1} (5)
    (1) edge node {$c$} (5)
    (2) edge node {$c^2$} (5)
    (3) edge node {$c^3$} (5)
    (4) edge node {$c^4$} (5)
    ;
  \end{tikzpicture}
  \caption{Reward ratio $b^{O(n)}$ for sophisticated agent}
  \label{fig:soph_reward_ratio}
\end{figure}
} 
\end{example}

Interestingly, and in contrast to sophisticated agents,
naive agents are doing quite well in reward settings. For the graph in
Figure \ref{fig:soph_reward_ratio} a naive agent will collect a reward
of $c^{n-1}$ (since it naively believes it is going to reach $v_n$
all the way until it gets to $v_{n-1}$).
More generally, we show \EC{in Appendix
\ref{app-naive-reward} }that in any graph the reward ratio for
a naive agent is at most $b$. \arxiv{
\input{app-rewards}
}

\subsection{Commitment Devices for Sophisticated Agents} \label{sec:commitment}

In example in Figure \ref{fig:soph_reward_ratio} above, we observed
that sophisticated agents can exhibit a high reward ratio. 
Since sophisticated agents are able to plan accurately about their
own biases, it is reasonable to consider ways they can modify
the instance to achieve better performance.
Such methods are referred to as commitment devices in the literature,
and they can be viewed as ways of committing one's future self
to a constrained course of action \cite{brocas2004commitment}. 
For example, a sophisticated individual
who believes they have a problem with spending too much might
put most of their money in a savings program that makes withdrawal
costly or difficult.

In this section we explore three ways of modeling commitment devices,
together with guarantees that apply to the set of all graphs.
We motivate them using the example in Figure \ref{fig:soph_reward_ratio}.

\textbf{1. Paying now to increase the reward later.} We note that
in Figure \ref{fig:soph_reward_ratio},
increasing the reward on the edge $(v_4,t)$ even by a small $\eps$
suffices to make the agent behave optimally and hence increase its
reward exponentially. The small increase in the reward on the edge
$(v_4,t)$ will cause the agent at $v_{3}$ prefer continuing to $v_4$. Given
this the agent at $v_2$ will continue to $v_3$ instead of going
directly to $t$ and so on. This means that the agent at $s$ will be
willing to pay quite a lot to increase the reward on the edge
$(v_4,t)$. Formalizing this idea we define a commitment device with a
``planning phase'' that allows the agent at $s$ to pay any cost $B$
and in return distribute an equivalent amount of reward along the
edges. Since the agent is present-biased the perceived cost of
spending this budget at the planning phase is $b \cdot B$ and not just $B$. In
Appendix \ref{app:planning} we analyze this commitment device and show
that (although it helps considerably for the example
in Figure \ref{fig:soph_reward_ratio})
there are instances in which even with this device the reward
ratio remains exponential.

\textbf{2. Adding bypasses.} 
In Figure \ref{fig:soph_reward_ratio},
consider adding
a direct edge of reward $0$ 
directly from $s$ to $v_4$. Adding such an edge will make the agent
traverse the path $s \rightarrow v_4 \rightarrow t$ and collect the
maximum reward. In Appendix \ref{app:zero-cost} we show that by adding
such edges we can guarantee that the reward of the agent in the new
graph will be at least $\frac{1}{b n}$ of the reward of the
optimal agent in the original graph.  This corresponds to an 
exponential improvement in the worst-case guarantee.
However, there are instances
in which adding zero-reward edges cannot
increase the agent's reward beyond this factor.

\textbf{3. Deleting edges.} 
Finally, we observe that in Figure \ref{fig:soph_reward_ratio},
if we delete the edge $(v_3,t)$ the agent will
traverse the optimal path. Using edge deletion to model commitment
devices is both powerful (as we will see next) and captures standard
commitment devices which are used in practice such as deadlines
(putting a deadline on some task ``deletes'' the option to do the task on
any time after the deadline). In the next subsection we show that by
deleting at most an $\eps$ fraction of the edges we can reduce the
reward ratio by an exponential amount in the worst case,
bringing the sophisticated agent's performance on the
modified graph to within a factor of $\Omega(n^{-\eps/2})$ of the
optimal in the original graph.


\subsection{Bounding Edge Removal}
Based on the final commitment device considered above --- deleting edges ---
we consider the following problem: Given a graph $G$, remove some fraction of the edges it to improve the reward that a sophisticated agent will collect. 
Here we show that by removing an arbitrarily small
constant fraction of the edges, 
we can increase the agent's reward in the modified graph to be within a polynomial factor of the optimal reward in the original graph:


\begin{theorem} \label{thm:edge-deletion}
For any DAG $G$ with $|V| > b$ and any $k>2$ we can remove $\frac{2|E|}{k}$ edges such that the reward of a sophisticated agent in the modified instance (denoted by $R_s'(s)$) is a factor of at least $n^{-k}$ times the optimal reward in the original instance (denoted by $R^*$). That is, $R_s'(s) \ge R^* n^{-k}$.
\end{theorem}

\begin{proof}
We begin with some intuition for the proof. The basic idea is that at
each vertex where the agent chooses to follow a path other than the
optimal path, the ratio between the perceived reward on the path the
agent takes and on the optimal path is at most $b$. This suggests that
by removing all edges with rewards in an interval $(b^j,b^{j+1})$ we
can guarantee that if the optimal path is of weight at least $b^{j+1}$,
then the sophisticated agent will not choose a path of weight less
than $b^{j}$. Stated this simply, the 
reasoning works for fan graphs as in
Figure \ref{fig:soph_reward_ratio}, but not in general, since
in fan graphs the fact that
the optimal path has a reward greater than $b^{j+1}$ in particular
implies that there is an edge with at least this reward. 
For more
general graphs, we need to be more careful and remove edges admitting
rewards in a larger interval to make sure that there are no paths with
rewards in this interval. This idea is formalized in the next claim

\begin{claim} \label{clm:int_bound}
  For any DAG $G$ with optimal reward $R^* = R_o(s)$ and $n = |V| > b$, if no
  edge outside of the optimal $s-t$ path $P_o$ has reward within the interval $(R^*n^{-j-1}, R^* n^{-j+1})$ for $j > 1$, then $R_s(s) \ge R^* n^{-j}$.
\end{claim}
\begin{proof}
  We name the vertices of the path $P_o$ by their order from $t$ to $s$. That is $u_0=t$, $u_1$ is the vertex closest to $t$ on the path and so on. We will show by induction on the vertices of the path $P_o$ that for each vertex $u_i \in P_o$, we have $R_s(u_i) \ge \min(R_o(u_i),R^*n^{-j+1}) - R^*(i-1)(b-1)n^{-j-1}$.
  
  \textbf{Base case:} At $t$, we have $R_s(t) \ge R_o(t) = 0$.
    
  \textbf{Inductive hypothesis:} For the vertex $u_{i-1}$ on $P_o$, we have
\\  $R_s(u_{i-1}) \ge \min(R_o(u_{i-1}),R^*n^{-j+1}) - R^*(i-1)(b-1)n^{-j-1}.$

  \textbf{Inductive step:} Consider the vertex $u_i$ on $P_o$. Let $v_s$ be the next vertex that the agent visits and let $v_o$ be the next vertex along $P_o$.

  If $r(u_i,v_s) \ge R^* n^{-j+1}$, then the claim is trivially true because
  $R_s(u_i) \ge r(u_i,v_s) \ge R^*n^{-j+1}$. Therefore, we will only consider
  $r(u_i,v_s) \le R^* n^{-j-1}$ (it cannot be between those values by assumption). 
  
  Because the agent chooses to go to $v_s$ instead of $v_o$, we have
  $b \cdot r(u_i,v_s) + R_s(v_s) \ge b \cdot  r(u_i,v_o) + R_s(v_o) \ge r(u_i, v_o) +
  \min(R_o(v_o), R^* n^{-j+1}) - R^*(i-1)(b-1)n^{-j-1}.$
  Subtracting $(b-1) r(u_i, v_s)$ from both sides and using the assumption that
    $r(u_i,v_s) \le R^* n^{-j-1}$, we have
    \begin{align*}
      r(u_i,v_s) + R_s(v_s) &\ge r(u_i,v_o) + \min(R_o(v_o), R^* n^{-j+1}) -
      R^*(i-1)(b-1)n^{-j-1} - (b-1)
      r(u_i,v_s) \\
      &\ge \min(r(u_i, v_o)+R_o(v_o), R^* n^{-j+1}) - R^*i(b-1)n^{-j-1}
    \end{align*}
    Thus, the inductive step holds.

%
  Finally, since $P_o$ has at most $n$ edges between $s$ and $t$, this means $R_s(s)
  \ge \min(R_o(s), R^* n^{-j+1}) - R^* n(b-1)n^{-j-1}= \min(R^*, R^* n^{-j+1}) - R^* n(b-1)n^{-j-1} =  R^*(n^{-j+1} - (b-1)
  n^{-j}) \ge R^* n^{-j}$.
\end{proof}

We can now use Claim \ref{clm:int_bound} to complete the proof of Theorem \ref{thm:edge-deletion}. Note that there are at most $|E|$ edges with reward in the interval $(R^*n^{-k-1},
  R^*)$, which we can divide into the overlapping intervals
 \\ $(R^*n^{-k-1},R^*n^{-k+1}), (R^*n^{-k},R^*n^{-k+2}),
  (R^*n^{-k+1},R^*n^{-k+3}), \dots, (R^*n^{-2},R^*).$ Each edge can fall into at
  most two such intervals. There are $k$ such intervals. Thus, the average
  interval has at most $\frac{2|E|}{k}$ edges in it. This means that there is
  some interval with at most $\frac{2|E|}{k}$ edges in it. Let this interval be
  $(R^*n^{-j-1}, R^*n^{-j+1})$ with $j \le k$. To produce the modified graph $G'$, we remove all edges from this interval that do not lie on the optimal path. Using Claim~\ref{clm:int_bound}, we know that if $G'$ has no edges outside the optimal path with reward in the interval $(R^*n^{-j-1}, R^*n^{j+1})$, then $R_s'(s) \ge R^* n^{-j} \ge R^* n^{-k}$.
\end{proof}

%% file: app-rewards.tex
First, we observe that the reward that a naive agent will collect starting at $u$ is:
$$R_n(u) = r(u, S_n^{(R)}(u)) + R_n(S_n^{(R)}(u)) \text{ where } S_n^{(R)}(u) =
\argmax_{v : (u,v) \in E} ~ b \cdot r(u, v) + R_o(v). $$
Next we show that the reward ratio for a naive agent is bounded by doing an induction over the node's height
\begin{claim} \label{clm:naive_reward}
 The reward ratio for a naive agent is at most $b$ ($b \cdot R_n(s) \ge R_o(s)$).
\end{claim}
\begin{proof}
  Consider some node $u$ from which the optimal agent goes to $v_o$ and the naive agent goes to $v_n$. Assume inductively that the claim holds for $v_n$ and $v_o$.
  Then, by using the induction assumption we have that :
  \begin{align*}
    b \cdot R_n(u) &= b \cdot c(u,v_n) + b \cdot R_n(v_n) \\
                   &\ge b \cdot c(u,v_n) + R_o(v_n)
  \end{align*}  
  Since the naive agent prefers $v_n$ over $v_o$ we have that  
    \begin{align*}
      b \cdot R_n(u) &\ge b \cdot c(u,v_n) + R_o(v_n) \\
				     & \ge b \cdot c(u,v_o) + R_o(v_o) \\
				     &\ge c(u,v_o) + R_o(v_o) = R_o(u)
    \end{align*}              
\end{proof}

%% file: conclusion.tex
In this work we incorporated the notion of sophistication
into a graph-theoretic model of planning, and thereby provided
a set of new guarantees about the performance of sophisticated
agents with present bias.
Our formalism makes it possible to identify basic new insights about
the behavior of such agents. These include the fact that 
their performance is at most a factor of $b$ worse than optimal;
that the worst-case performance of partially naive agents displays
a sharp transition at the boundary between optimism and pessimism;
that rewards for motivating sophisticated agents can display
a surprising form of non-monotone behavior;
and that any instance for a reward-seeking sophisticated
agent with a bad performance ratio
can be dramatically improved by deleting small
fraction of the instance.

Our work leaves open a number of interesting questions.
We have identified some specific questions in the text thus far,
and highlight the following additional ones here.
First, can we say anything about the average-case performance
of sophisticated and partially naive agents on ensembles of
instances defined by a natural family of random graphs?
It is possible that particular models might exhibit contrasts
at the level of average-case performance even though they do not
at the level of worst-case performance.
Second, can we characterize the set of paths traversed by a 
sophisticated agent on a given graph as we vary the value of $b$?
Such a characterization might yield a bound on the number of possible
paths over all values of $b$, or allow us to analyze agents whose
belief about their value of $b$ is drawn from a distribution
rather than pinned to a specific value.
And finally, it would be interesting to understand the power of
commitment devices more fully, both through the development of
other plausible families of commitment devices, and through tighter
bounds on the ones studied here.

%% file: partial-naive.tex
A partially naive agent is an agent that has present bias $b$ but
believes it has a bias $b'$. In the case where $b = b'$, the agent is fully
sophisticated. In the case where $b' = 1$, the agent is naive.

Let $C_p(u)$ be the cost for a partially naive agent with parameters
$b$ and $b'$ to go from $u$ to the goal $t$. Let $C_s'(u)$ be the cost for a
fully sophisticated agent with parameter $b'$ to go from $u$ to the goal $t$.
The cost of the path that a partially naive agent  takes is:
\begin{equation*} 
  C_p(u) = c(u, S_p(u)) + C_p(S_p(u)) \text{ where } S_p(u) = \argmin_{v : (u,v)
  \in E} b \cdot c(u, v) + C_s'(v).
\end{equation*}

We will consider two types of partially sophisticated agents: optimistic ($b' <
b$) and pessimistic ($b' > b$). 

\subsection{Optimistic Partially Naive}
The path taken by an optimistic partially naive agent can have cost
ratio $b^{O(n)}$ where $n$ is the number of vertices. In the graph shown in
Figure~\ref{fig:optpart}, at each vertex $v_i$, the agent has a choice between
going directly to $t$, which has a perceived cost of $b^{i+1}$, or proceeding to
$v_{i+1}$. Since the agent believes it has bias $b'$, it believes that at
$v_{i+1}$, it will proceed to $t$ since $b' b^{i+1} < b^{i+2}$. Therefore, the
perceived cost of the path $v_i \to v_{i+1} \to t$ is $b^{i+1}$. If we choose
the tie breaking scheme to favor $v_{i+1}$ over $v_i$ (or subtract $\varepsilon$
from the weights as necessary), we can get the agent to follow the path $s \to
v_1 \to v_2 \to v_3 \to v_4 \to t$, which has cost $b^4$, instead of the optimal
path $s \to t$ which has cost 1. It is easy to see that the same analysis on a fan of size $n$ would show that its cost ratio is $b^{O(n)}$.
\begin{figure}[ht]
  \centering
  \begin{tikzpicture}[->,shorten >=1pt,auto,node distance=2cm, thin]
    \node (0) [circle, draw, minimum size=1cm] at (0,0) {$s$};
    \node (1) [circle, draw, minimum size=1cm] at (1,2) {$v_1$};
    \node (2) [circle, draw, minimum size=1cm] at (3,3) {$v_2$};
    \node (3) [circle, draw, minimum size=1cm] at (5,2) {$v_3$};
    \node (4) [circle, draw, minimum size=1cm] at (6,0) {$v_4$};
    \node (5) [circle, draw, minimum size=1cm] at (3,0) {$t$};
    
    \path
    (0) edge node {0} (1)
    (1) edge node {0} (2)
    (2) edge node {0} (3)
    (3) edge node {0} (4)

    (0) edge node {1} (5)
    (1) edge node {$b$} (5)
    (2) edge node {$b^2$} (5)
    (3) edge node {$b^3$} (5)
    (4) edge node {$b^4$} (5)
    ;
  \end{tikzpicture}
  \caption{Cost ratio $b^{O(n)}$ for optimistic partially sophisticated agent}
  \label{fig:optpart}
\end{figure}
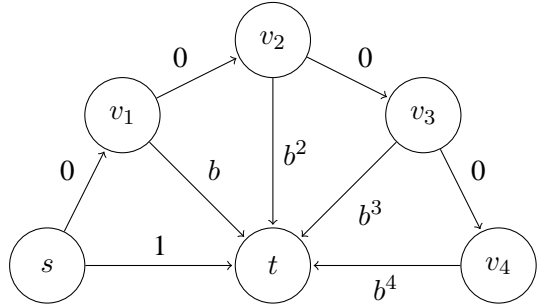
\subsection{Pessimistic Partially Sophisticated}
\begin{claim}
  The path taken by a pessimistic partially sophisticated agent can have cost
  ratio at most $b'$.
\end{claim}
\begin{proof}
  We use induction on the height of the vertex to show that  $C_p(u) \le C_s'(u)$ for every node $u$ in the graph. Recall that $C_s'(u) \le b' \cdot C_o(u)$ by Theorem~\ref{thm:full}. As this holds for $s$ in particular we have that $C_p(s) \le b' \cdot C_o(s)$, as required. We now provide the inductive argument showing that $C_p(u) \le C_s'(u)$.
 
  \textbf{Base case:} $C_p(t) = 0$, $C_s'(t) = 0$.

  \textbf{Inductive hypothesis:} For a vertex $u$ of height $k$, $C_p(u) \le
  C_s'(u)$.

  \textbf{Inductive step:} For a vertex $u$ of height $k+1$,
  let $v_p = S_p(u)$ and $v_{s'} = S_s'(u)$. 
  Since the partially sophisticated agent chose to go to $v_p$ and the fully sophisticated
  agent chose to go to $v_{s'}$, we know
    \begin{align*}
      b\cdot c(u,v_p) + C_s'(v_p) &\le b\cdot c(u,v_{s'}) + C_s'(v_{s'}) \\
      b' \cdot c(u,v_p) + C_s'(v_p) &\ge b' \cdot c(u,v_{s'}) + C_s'(v_{s'})
    \end{align*}
  Subtracting one equation from the other, we have $ (b'-b) \cdot c(u,v_p) \ge (b'-b) \cdot c(u,v_{s'})$ implying that $c(u,v_p) \ge c(u,v_{s'})$.
  
  Since $ b\cdot c(u,v_p) + C_s'(v_p) \le b\cdot c(u,v_{s'}) + C_s'(v_{s'})$ and $c(u,v_p) \ge c(u,v_{s'})$, we have that $ c(u,v_p) + C_s'(v_p) \le c(u,v_{s'}) + C_s'(v_{s'})$. By the induction hypothesis, we have that $C_p(v_{p}) \le
  C_s'(v_{p})$. This implies that $c(u,v_p) + C_p(v_p) \le C_s'(u)$, and hence $C_p(u) \le
  C_s'(u)$.
\end{proof}

The example in Figure~\ref{fig:pessimist} shows that this bound is tight.
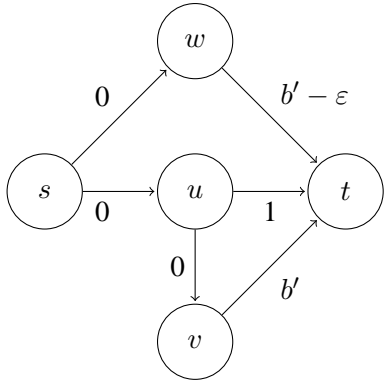
\begin{figure}[ht]
  \centering
  \begin{tikzpicture}[->,shorten >=1pt,auto,node distance=2cm, thin]
    \node (s) [circle, draw, minimum size=1cm] at (0,0) {$s$};
    \node (u) [circle, draw, minimum size=1cm] at (2,0) {$u$};
    \node (v) [circle, draw, minimum size=1cm] at (2,-2) {$v$};
    \node (w) [circle, draw, minimum size=1cm] at (2,2) {$w$};
    \node (t) [circle, draw, minimum size=1cm] at (4,0) {$t$};

    \path
    (s) edge node [below left] {0} (u)
    (s) edge node {0} (w)
    (u) edge node [below] {1} (t)
    (u) edge node [left] {0} (v)
    (v) edge node [below right] {$b'$} (t)
    (w) edge node {$b' - \varepsilon$} (t)
    ;
  \end{tikzpicture}
  \caption{Worst case scenario for a pessimistic partially sophisticated agent}
  \label{fig:pessimist}
\end{figure}

\subsection{Naive Future-biased Agent}
The following result is about naive agents and not sophisticated ones. However, since the proof follows a similar path to the proofs for the bounds on the cost ratio for sophisticated agents we include it in this paper.

A future-biased agent is a naive agent that has bias $b < 1$. Let $C_f$ be the perceived cost function
for a future-biased agent. The decision rule is the same for any naive agent:
$$S_f(u) = \argmin_{v : (u,v) \in E} b \cdot c(u, v) + C_o(v)$$
\begin{claim} \label{clm:fut}
  The cost ratio for a future-biased agent is at most $\frac{1}{b}$.
\end{claim}
\begin{proof}
  We proceed by induction on the height of the vertex.

  \textbf{Base case:} $C_f(t) = 0$, $C_o(t) = 0$.

  \textbf{Inductive hypothesis:} For a vertex $u$ of height $k$, $C_f(u) \le
  \frac{1}{b} C_o(u)$.

  \textbf{Inductive step:} Consider a vertex $u$ of height $k+1$. Let $v_f =
  S_f(u)$ and $v_o = S_o(u)$. Since the future-biased agent chose to go to $v_f$, we have that $b \cdot c(u,v_f) + C_o(v_f) \le b \cdot c(u,v_o) + C_o(v_o)$. Using the
  induction hypothesis, which says that $C_o(v_f) \ge b \cdot C_f(v_f)$, and because
  $b < 1$, we have that $b \cdot c(u,v_f) + b \cdot C_f(v_f) \le C_o(u)$ which
  implies that $b \cdot C_f(u) \le C_o(u)$ as required.
\end{proof}

\noindent
The example in Figure~\ref{fig:future} shows that this bound is tight.
\begin{figure}[ht]
  \centering
  \begin{tikzpicture}[->,shorten >=1pt,auto,node distance=2cm, thin]
    \node (s) [circle, draw, minimum size=1cm] at (0,0) {$s$};
    \node (v) [circle, draw, minimum size=1cm] at (1.5,1.5) {$v$};
    \node (t) [circle, draw, minimum size=1cm] at (3,0) {$t$};

    \path
    (s) edge node [above] {$\frac{1}{b}$} (t)
    (s) edge node [above left] {0} (v)
    (v) edge node [above right] {1} (t)
    ;
  \end{tikzpicture}
  \caption{Worst case scenario for a future-biased agent}
  \label{fig:future}
\end{figure}
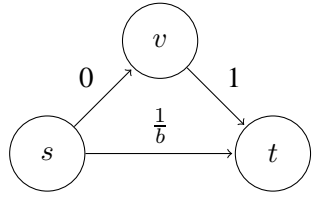

%% file: traversable.tex
\EC{
\subsection{Upper Bound for Traversable Rewards} \label{app:rewardratio}
We prove
Theorem~\ref{thm:rewardratio}, which states that the ratio between the
minimal reward needed at $t$ to motivate an agent to traverse $G$ and $C_o(s)$
is at most $b$.
\input{ub-traversable}
}

\subsection{Exponential Traversable Paths} \label{app:exp-paths}
Here, we provide a proof that the graph $G$ shown in Figure~\ref{fig:exp_paths}
has exponentially many distinct paths taken by the agent for various values of
the reward placed at $t$. Recall that for any path in $G$, we define the
bit string $x$ whose $i$th bit is 0 if the path contains the direct $v_i \to
v_{i+1}$ edge and 1 if the path contains the edges $v_i \to w_i \to v_{i+1}$.

Number the bits representing $x$ from 0 (least significant bit) to $n-1$ (most
significant bit). For a number $x$ such that $0 \le x \le 2^n - 1$, let $Q = \{i
~ | ~ \text{the $i$th bit of $x$ is a 1}\}$. Proceeding by induction from $n-1$
to 0, we will show that for a reward of $2^n + (bc-2)x$, if $i \in Q$, then
agent would be willing to go through $w_i$, and if $i \notin Q$, then the agent
would be willing to go directly to $v_{i+1}$ from $v_i$ but not through $w_i$.

Let $x_i$ be the $i$-truncated version of $x$, meaning that all bits strictly
less significant than the $i$th bit are set to 0.

\textbf{Case 1:} $i \in Q$. Then, the perceived cost of going through $w_i$ is
\begin{align*}
  b \cdot c \cdot 2^i + \sum_{i<j<n,j \in Q} c \cdot 2^j + \sum_{i<j<n,j \notin
  Q} 2^j
  &= b \cdot c \cdot 2^i + \sum_{i<j<n} 2^j + \sum_{i<j<n,j \in Q} (c-1) 2^j \\
  &= b \cdot c \cdot 2^i + 2^n - 2^{i+1} + (c-1)x_{i+1} \\
  &= 2^n + (bc-2) 2^i + (bc-2)x_{i+1} \tag{because $1+c=bc$} \\
  &= 2^n + (bc-2) x_i \\
  &\le 2^n + (bc-2) x
\end{align*}
Therefore, if $i \in Q$, then the agent is willing to go to $w_i$. Because $c
\le b$ (and we can break ties by always choosing $w_i$ over $v_{i+1}$), at
$v_i$, the agent will choose to go to $w_i$.

\textbf{Case 2:} $i \notin Q$. Then, the perceived cost of going through $w_i$
is again $2^n + (bc-2) 2^i + (bc-2)x_{i+1}$. However,
\begin{align*}
  x &= \sum_{0 \le j \le i, j \in Q} 2^j + x_{i+1} \\
    &= \sum_{0 \le j < i, j \in Q} 2^j + x_{i+1} \\
    &\le \sum_{j=0}^{i-1} 2^j + x_{i+1} \\
    &< 2^i + x_{i+1}
\end{align*}
Therefore, the perceived cost of going through $w_i$ is
$$2^n + (bc-2)(2^i+x_{i+1}) > 2^n + (bc-2) x$$
and thus, the agent would not be willing to go through $w_i$. The perceived cost
of going directly from $v_i$ to $v_{i+1}$, however, is
\begin{align*}
  b \cdot 2^i + 2^n - 2^{i+1} + (c-1)x_{i+1}
  &= 2^n + (b-2) 2^i + (bc-2) x_{i+1} \\
  &< 2^n + (bc-2) x \tag{because $b < 2$}
\end{align*}

In either case, the agent follows the path corresponding to the binary
representation of $x$. Thus, by induction, for a reward of $2^n + (bc-2) x$, the
agent must follow this path.

Since every number between $0$ and $2^n-1$ can be represented by paths in this
graph, under these conditions, there are an exponential number of different
paths that the agent takes as the reward varies.

It is interesting to note that the only way we can have $b = c$ is to set them
both to the golden ratio $\phi$.

\EC{
\subsection{Edge Deletion} \label{app:deletion}
Proof of Claim~\ref{clm:single-path}, which states that the minimal motivating
subgraph of a graph $G$ is a single path.
\input{edge-deletion}
}

\subsection{Ratio between Internal Rewards and Edge Deletion}
\label{app:internal-ratio}

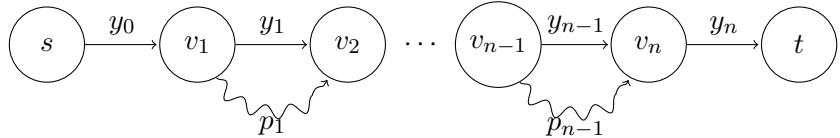
\begin{figure}[ht]
  \centering
  \begin{tikzpicture}[->,shorten >=1pt,auto,node distance=2cm, thin]
    \node (s) [circle, draw, minimum size=1cm] at (0,0) {$s$};
    \node (v1) [circle, draw, minimum size=1cm] at (2,0) {$v_1$};
    \node (v2) [circle, draw, minimum size=1cm] at (4,0) {$v_2$};
    \node (d) at (5,0) {$\cdots$};
    \node (vn1) [circle, draw, minimum size=1cm] at (6,0) {$v_{n-1}$};
    \node (vn) [circle, draw, minimum size=1cm] at (8,0) {$v_n$};
    \node (t) [circle, draw, minimum size=1cm] at (10,0) {$t$};

    \path
    (s) edge node {$y_0$} (v1)
    (v1) edge node {$y_1$} (v2)
    (vn1) edge node {$y_{n-1}$} (vn)
    (vn) edge node {$y_n$} (t)
    ;

    \path [->,decoration=snake]
    (v1) edge [decorate,bend right=60] node [below] {$p_1$} (v2)
    (vn1) edge [decorate,bend right=60] node [below] {$p_{n-1}$} (vn)
    ;
  \end{tikzpicture}
  \caption{$G$ such that $2R_d \approx R_i$}
  \label{fig:internal-ratio-redux}
\end{figure}

Consider the graph in Figure~\ref{fig:internal-ratio-redux}. Let $R_d$ be the minimal reward
required if the agent can delete edges from the graph, and let $R_i$ be the
minimal reward required if the agent can distribute the reward internally. Then,
we make the following definitions:
\begin{align*}
  y_i &= \p{\frac{b}{b-1}}^i \\
  p_i &= y_i + \frac{p_{i-1}}{2} \\
  p_1 &= \frac{(b+1)y_i}{2}
\end{align*}
The curved paths represent paths made up by a seriew of low-weight edges, such
that the perceived cost is equal (within a factor of $\varepsilon$) to the true
cost. From each $v_i$, the perceived cost of taking the direct edge to $v_{i+1}$
is $b y_i$, while the perceived cost of taking the longer path is $p_i$. If
deletion is allowed, then from each $v_i$ (and $s$), the perceived cost of the
optimal path with deletion is $R_d = b\p{\frac{b}{b-1}}^n$. However, by only
allowing internal reward placement, since the perceived cost of following the
$(v_n, t)$ edge is $b\p{\frac{b}{b-1}}^n$, meaning that the agent must place
that much reward just on the last edge. Consider each vertex inductively
backwards from $v_n$ -- in order to motivate the agent to get from $v_i$ to
$v_{i+1}$, we can either place enough reward on the direct edge to bring the
perceived cost down below $p_i$, or we can allow the agent to take the longer
bypass sequence of edges. However, if it takes the longer sequence, this adds
$p_i - y_i$ to the total cost to the path, meaning that the overall reward
required is increased by $p_i - y_i$. Since the reward required to bring the
perceived cost of following the optimal path down to $p_i$ is $b y_i - p_i$, we
choose $p_i$ so as to make the agent indifferent to both choices. Let $D_i$ be
the amount of additional reward needed to motivate the agent to get past $v_i$.
Then, we have $D_i = b y_i - p_i = p_i - y_i + D_{i-1} = \sum_{j=1}^i (p_i -
y_i)$.

First, we observe that if we place reward of $b y_i - p_i$ on the direct $(v_i,
v_{i+1})$ edge, then this is enough to motivate the agent to get past $v_i$.
This is because if the agent follows every bypass before $v_i$, it incurs an
additional cost of $\sum_{j=1}^{i-1} (p_i - y_i) = D_{i-1}$. Since by definition
$b y_i - p_i = D_i  = p_i - y_i + D_{i-1}$, the additional reward placed on the
optimal $(v_i, v_{i-1})$ edge is enough to motivate the agent to take all the
bypasses up to $v_{i-1}$. However, if instead we place a reward of $p_i - y_i$
on the bypass from $v_i$ to $v_{i+1}$, from any vertex before $v_i$, the
perceived cost of the path is the same as it would be if the agent took the
optimal $(v_i, v_{i+1})$ edge. Thus, inductively, the agent still has to be
motivated to get to $v_{i-1}$, which by induction takes an additional reward of
$D_{i-1}$. Therefore, we want the agent to be indifferent between motivating the
direct edge and motivating the bypass, meaning
\begin{align*}
  b y_i - p_i &= p_i - y_i + D_{i-1} \\
  (b+1) y_i &= 2p_i + b y_{i-1} - p_{i-1} \\
  (b+1) y_i &= 2p_i + b \p{\frac{b-1}{b}} y_i - p_{i-1} \\
  y_i \p{b+1 - (b-1)} + p_{i-1} &= 2p_i \\
  y_i + \frac{p_{i-1}}{2} &= p_i
\end{align*}
which matches the above definition. Furthermore,
\begin{align*}
  b y_1 - p_1 &= p_1 - y_1 \\
  (b+1) y_1 &= 2p_1 \\
  \frac{(b+1)y_1}{2} &= p_1
\end{align*}
Thus, we can explicitly solve for $p_i$:
\begin{align*}
  p_i &= \frac{(b+1)b}{(b-1)2^i} + \sum_{j=2}^i \p{\frac{b}{b-1}}^j 2^{j-i} \\
      &= \frac{(b+1)b}{(b-1)2^i} + 2^{-i} \p{\frac{2b}{b-1}} \sum_{j=1}^{i-1}
      \p{\frac{2b}{b-1}}^j \\
      &= \frac{(b+1)b}{(b-1)2^i} + 2^{-i} \p{\frac{2b}{b-1}}^2
      \p{\frac{\p{\frac{2b}{b-1}}^{i-1}-1}{\frac{2b}{b-1}-1}} \\
      &= \frac{(b+1)b}{(b-1)2^i} + 2^{-i+2} \p{\frac{b}{b-1}}^2
      \p{\frac{\p{\frac{2b}{b-1}}^{i-1}-1}{\frac{b+1}{b-1}}} \\
      &= \frac{(b+1)b}{(b-1)2^i} + \p{\frac{b^2}{2^{i-2}(b+1)(b-1)}}
      \p{\p{\frac{2b}{b-1}}^{i-1}-1} \\
\end{align*}
Using this, we have
\begin{align*}
  D_i &= b y_i - p_i \\
      &= b\p{\frac{b}{b-1}}^i - \frac{(b+1)b}{(b-1)2^i} -
      \p{\frac{b^2}{2^{i-2}(b+1)(b-1)}} \p{\p{\frac{2b}{b-1}}^{i-1}-1} \\
      &= \p{\frac{b}{b-1}}^i \p{b - \frac{2b}{b+1}} + \frac{b}{2^i}
      \p{\frac{4b}{b^2-1} - \frac{b+1}{b-1}} \\
      &= \p{\frac{b}{b-1}}^i \p{\frac{b^2-b}{b+1}} + \frac{b}{2^i}
      \p{\frac{-b^2+2b-1}{b^2-1}} \\
      &= \p{\frac{b}{b-1}}^i \p{\frac{b(b-1)}{b+1}} - \frac{b}{2^i}
      \p{\frac{b-1}{b+1}} \\
      &= \frac{b(b-1)}{b+1} \p{\p{\frac{b}{b-1}}^i - \frac{1}{2^i}}
\end{align*}
Since $D_{n-1}$ is the additional reward required to get the agent past the
first $n-1$ vertices, we have $R_i = b\p{\frac{b}{b-1}}^n + D_{n-1}$. This makes
the ratio
\begin{align*}
  \frac{R_i}{R_d} &= \frac{b\p{\frac{b}{b-1}}^n + D_{n-1}}{b\p{\frac{b}{b-1}}^n}
  \\
  &= 1 + \frac{\frac{b(b-1)}{b+1}\p{\p{\frac{b}{b-1}}^{n-1} -
  \frac{1}{2^{n-1}}}}{b\p{\frac{b}{b-1}}^n} \\
  &= 1 + \p{\frac{b-1}{b+1}}\p{\frac{b-1}{b}} - 2\p{\frac{b-1}{2b}}^n \\
  &= 1 + \p{\frac{(b-1)^2}{b(b+1)}} - 2\p{\frac{b-1}{2b}}^n
\end{align*}
For large enough $b$ and $n$, this ratio approaches 2.

%% file: commitment.tex
\subsection{Planning Agents} \label{app:planning}

We consider the ``planning agent'' described in Section~\ref{sec:commitment}.
First, we observe that the agent would never place on a reward on an edge that
it would not traverse. Since the agent is sophisticated, it knows exactly which
path it will take through a graph, and thus the path would not be affected by
increasing the reward on an edge that the agent knows it will not take.

Let $R_o(s)$ and $R_s(s)$ denote the rewards achieved by the optimal and
sophisticated agents respectively through the graph. Let $R_p(s)$ denote the
reward achieved by a sophisticated agent which is allowed to have this planning
phase. Let $B$ be the total reward that the planning agent distributes along the
edges.

\begin{claim}
  The amount of additional reward a planning agent can add to the graph is at
  most the difference between the optimal reward and the reward received by a
  standard sophisticated agent divided by $b-1$. In other words,
  $$B \le \frac{R_o(s) - R_s(s)}{b-1}$$
\end{claim}
\begin{proof}
  Since the perceived cost of taking on a cost of $B$ during the planning phase
  is $b \cdot B$, we know that the planning agent must get a reward of at least
  $R_s(s) + b \cdot B$. However, we know that $R_p(s) \le R_o(s) + B$, since the
  best the planning agent can do is follow the optimal path and reclaim the
  reward of $B$ that it put along the edges. Thus,
  \begin{align*}
    R_o(s) + B &\ge R_p(s) \\
               &\ge R_s(s) + b \cdot B \\
    R_o(s) - R_s(s) &\ge (b-1)B \\
    \frac{R_o(s) - R_s(s)}{b-1} &\ge B
  \end{align*}
\end{proof}

\begin{claim} \label{clm:planning-bad}
  There are graphs for which a planning agent has an exponentially bad reward
  ratio.
\end{claim}
\begin{proof}
  Consider the graph in Figure~\ref{fig:planning_bad}, with $x_i =
  \p{\frac{b(b-1)}{2b-1}}^i$. We will show that it is impossible for a planning
  agent to take any path other than $s \to t$.

  Assume towards contradiction that there is some $k$ such that the planning
  agent can take the path $s \to v_1 \to \dots \to v_k \to t$. This means that
  the agent will receive a reward of $x_k+B$. However,
  since the perceived cost of adding a reward of $B$ to the graph is $b \cdot B$
  during the planning phase, we know
  \begin{align*}
    R_s(s) + b \cdot B &\le x_k + B \\
    B &\le \frac{x_k-1}{b-1} \tag{$R_s(s)=1$}
  \end{align*}
  However, in order to motivate the agent to go from $v_{k-1}$ to $v_k$, the
  perceived reward of $v_{k-1} \to t$ must be less than the perceived reward of
  $v_{k-1} \to v_k \to t$. Let $B_k$ be the reward placed on the $v_{k-1} \to
  v_k$ edge. This means
  \begin{align*}
    b x_{k-1} &\le b \cdot B_k + x_k \\
              &\le b \cdot B + x_k \\
              &\le b \p{\frac{x_k-1}{b-1}} + x_k \\
              &= \p{\frac{2b-1}{b-1}}x_k - \frac{b}{b-1} \\
    x_{k-1} &\le \p{\frac{2b-1}{b(b-1)}} x_k - \frac{b}{b-1}
  \end{align*}
  This is a contradiction because by definition, $x_{k-1} =
  \p{\frac{2b-1}{b(b-1)}} x_k$. Thus, the agent cannot reach $x_k$ for any $k
  \ge 1$, so it can only take the path $s \to t$.

\figureswraper{
  \begin{figure}[ht]
    \centering
    \begin{tikzpicture}[->,shorten >=1pt,auto,node distance=2cm, thin]
      \node (0) [circle, draw, minimum size=1cm] at (0,0) {$s$};
      \node (1) [circle, draw, minimum size=1cm] at (1,2) {$v_1$};
      \node (2) [circle, draw, minimum size=1cm] at (3,3) {$v_2$};
      \node (3) [circle, draw, minimum size=1cm] at (5,2) {$v_3$};
      \node (4) [circle, draw, minimum size=1cm] at (6,0) {$v_4$};
      \node (5) [circle, draw, minimum size=1cm] at (3,0) {$t$};
      
      \path
      (0) edge node {0} (1)
      (1) edge node {0} (2)
      (2) edge node {0} (3)
      (3) edge node {0} (4)

      (0) edge node {1} (5)
      (1) edge node {$x_1$} (5)
      (2) edge node {$x_2$} (5)
      (3) edge node {$x_3$} (5)
      (4) edge node {$x_4$} (5)
      ;
    \end{tikzpicture}
    \caption{Exponentially bad reward ratio for planning agent.}
    \label{fig:planning_bad}
  \end{figure}
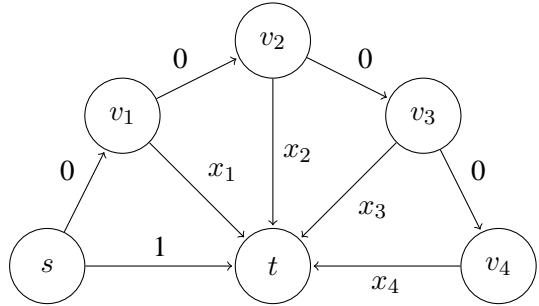
}
\end{proof}

\subsection{Zero-Reward Edges} \label{app:zero-cost}
We consider the second commitment device described in
Section~\ref{sec:commitment}, by which an agent can add edges of reward zero to
the graph.

\begin{claim}
  The reward ratio for an agent that can add zero-reward edges is at most 
  $bn$, where $n$ is the number of nodes in the graph. In other words, $$bn
  \cdot R_z(s) \ge R_o(s)$$
\end{claim}
\begin{proof}
  Let $R_z(s)$ be the maximum reward attainable by a sophisticated agent allowed
  to add zero-reward edges to the graph. Assume that the reward ratio without
  adding any zero-reward edges is larger than than $bn$, meaning that $bn
  R_s(s) < R_o(s)$ (if not, then the claim holds trivially). Since there are $n$
  vertices, the path taken by the optimal agent must have length at most $n$.
  Therefore, there must be some edge on that path with reward at least
  $\frac{R_o(s)}{n} > b \cdot R_s(s)$. This means that if the zero-reward agent
  were to add an edge of reward 0 from $s$ to the beginning of the edge of
  maximum reward along the sophisticated agent's path, the agent would prefer to
  take the zero reward edge, as the perceived reward of doing so is more than $b
  R_s(s)$, while the perceived reward of not taking the zero-reward edge and
  following the rest of the sophisticated agent's path is at most $b R_s(s)$.
  Thus, the zero-reward agent would achieve a reward of at least
  $\frac{R_o(s)}{n}$, meaning its reward ratio would be smaller than
  $bn$.
\end{proof}

Figure~\ref{fig:zero_bad} shows that it's possible for the reward ratio between
the optimal agent and the sophisticated agent to be linear and for zero-reward
edges to have no benefit. The optimal agent gets a reward of $n(b-1)$, but the
sophisticated agent gets a reward of 1 (assuming ties are broken by choosing the
largest outgoing edge). Furthermore, there is no way to add zero-reward edges to
the graph to increase the reward received by the sophisticated agent.

\begin{figure}[ht]
  \centering
  \begin{tikzpicture}[->,shorten >=1pt,auto,node distance=2cm, thin]
    \node (0) [circle, draw, minimum size=1cm] at (0,0) {$s$};
    \node (11) [circle, draw, minimum size=1cm] at (-1,2) {$v_1'$};
    \node (1) [circle, draw, minimum size=1cm] at (1,2) {$v_1$};
    \node (22) [circle, draw, minimum size=1cm] at (1,4) {$v_2'$};
    \node (2) [circle, draw, minimum size=1cm] at (3,3) {$v_2$};
    \node (33) [circle, draw, minimum size=1cm] at (5,4) {$v_3'$};
    \node (3) [circle, draw, minimum size=1cm] at (5,2) {$v_3$};
    \node (44) [circle, draw, minimum size=1cm] at (7,2) {$v_4'$};
    \node (4) [circle, draw, minimum size=1cm] at (6,0) {$v_4$};
    \node (5) [circle, draw, minimum size=1cm] at (3,0) {$t$};
    
    \path
    (0) edge node {0} (11)
    (11) edge node {$b-1$} (1)
    (1) edge node {0} (22)
    (22) edge node {$b-1$} (2)
    (2) edge node {0} (33)
    (33) edge node {$b-1$} (3)
    (3) edge node {0} (44)
    (44) edge node {$b-1$} (4)

    (0) edge node {1} (5)
    (1) edge node {1} (5)
    (2) edge node {1} (5)
    (3) edge node {1} (5)
    (4) edge node {1} (5)
    ;
  \end{tikzpicture}
  \caption{No benefit to adding zero-reward edges.}
  \label{fig:zero_bad}
\end{figure}
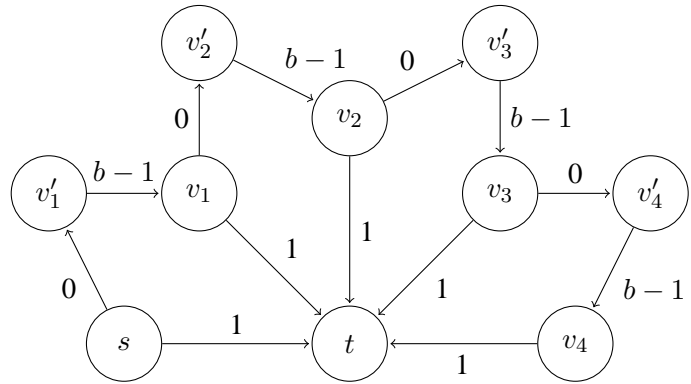

%% file: soph-arxiv.bbl
\begin{thebibliography}{10}

\bibitem{akerlof-procrastination}
George~A. Akerlof.
\newblock Procrastination and obedience.
\newblock {\em American Economic Review: Papers and Proceedings}, 81(2):1--19,
  May 1991.

\bibitem{brocas2004commitment}
Isabelle Brocas, Juan~D Carrillo, and Mathias Dewatripont.
\newblock Commitment devices under self-control problems: An overview.
\newblock {\em The Psychology of economic decisions}, 2:49--67, 2004.

\bibitem{dellavigna2007psychology}
Stefano DellaVigna.
\newblock Psychology and economics: Evidence from the field.
\newblock Technical report, National Bureau of Economic Research, 2007.

\bibitem{dellavigna2006paying}
Stefano DellaVigna and Ulrike Malmendier.
\newblock Paying not to go to the gym.
\newblock {\em The American Economic Review}, pages 694--719, 2006.

\bibitem{frederick-time-inconsist-surv}
Shane Frederick, George Loewenstein, and Ted O'Donoghue.
\newblock Time discounting and time preference.
\newblock {\em Journal of Economic Literature}, 40(2):351--401, June 2002.

\bibitem{ko-ec14}
Jon Kleinberg and Sigal Oren.
\newblock Time-inconsistent planning: a computational problem in behavioral
  economics.
\newblock In {\em {ACM} Conf. on Economics and Computation}, 2014.

\bibitem{laibson-quasi-hyperbolic}
David Laibson.
\newblock Golden eggs and hyperbolic discounting.
\newblock {\em Quarterly Journal of Economics}, 112(2):443--478, 1997.

\bibitem{odonoghue-now-or-later}
Ted O'Donoghue and Matthew Rabin.
\newblock Doing it now or later.
\newblock {\em American Economic Review}, 89(1):103--124, March 1999.

\bibitem{odonoghue-choice-procrastination}
Ted O'Donoghue and Matthew Rabin.
\newblock Choice and procrastination.
\newblock {\em Quarterly Journal of Economics}, 116(1):121--160, February 2001.

\bibitem{pollak-time-inconsist}
R.~A. Pollak.
\newblock Consistent planning.
\newblock {\em Review of Economic Studies}, 35(2), 1968.

\bibitem{strotz-time-inconsist}
R.~H. Strotz.
\newblock Myopia and inconsistency in dynamic utility maximization.
\newblock {\em Review of Economic Studies}, 23(3):165--180, 1955.

\end{thebibliography}
